\documentclass[a4paper,12pt,reqno,superscriptaddress,showkeys,nofootinbib]{revtex4}
\usepackage[centertags]{amsmath}
\usepackage{amsfonts}
\usepackage{amssymb}
\usepackage{amsthm}
\usepackage{newlfont}
\usepackage{stmaryrd}
\usepackage{mathrsfs}
\usepackage{euscript}
\usepackage{siunitx}
\usepackage{physics}
\usepackage{algorithm2e}
\usepackage{graphicx,subcaption}
\usepackage{enumitem}
\usepackage{natbib}
\usepackage{graphicx}
\usepackage{color}
\usepackage{floatrow}
\usepackage{caption}
\usepackage{hyperref}
\usepackage{hhline}
\usepackage{hyperref}
\usepackage{cleveref}

\theoremstyle{plain}
  \newtheorem{theorem}{Theorem}[section]
  \newtheorem{corollary}[theorem]{Corollary}
  \newtheorem{proposition}[theorem]{Proposition}
  \newtheorem{lemma}[theorem]{Lemma}
\theoremstyle{definition}
  \newtheorem{definition}{Definition}[section]

\theoremstyle{remark}

  \newtheorem{example}[theorem]{Example}

\numberwithin{equation}{section}

 \let\be=\beta  
\let\ve=\varepsilon   
 \let\la=\lambda

\newcommand{\caV}{{\mathcal V}}

\newcommand{\bbR}{{\mathbb R}}

\newcommand{\opunit}{\text{1}\kern-0.22em\text{l}}

\DeclareMathAlphabet{\mathpzc}{OT1}{pzc}{m}{it}

\newcommand{\id}{\textrm{d}}

\usepackage{floatrow}
\usepackage{caption}

\begin{document}
\title{The vanishing of excess heat for\\nonequilibrium processes reaching zero ambient temperature}
\author{Faezeh Khodabandehlou}
	\affiliation{Instituut voor Theoretische Fysica, KU Leuven, Belgium}
	\author{Christian Maes}
\affiliation{Instituut voor Theoretische Fysica, KU Leuven, Belgium}
\author{Irene Maes}
	\affiliation{Departement Wiskunde, KU Leuven, Belgium}	%
		\author{Karel Neto\v{c}n\'{y}}
		\affiliation{Institute of Physics, Czech Academy of Sciences, Prague, Czech Republic}
	\keywords{}
	
\keywords{Nernst postulate, excess heat, nonequilibrium heat capacity, matrix-forest theorem.}

\begin{abstract}
We present the mathematical ingredients for an extension of the Third Law of Thermodynamics (Nernst heat postulate) to nonequilibrium processes.  The central quantity is the excess heat which measures the quasistatic addition to the steady dissipative power when a parameter in the dynamics is changed slowly.  We prove for a class of driven Markov jump processes that it vanishes at zero environment temperature.  Furthermore, the nonequilibrium  heat capacity goes to zero with temperature as well.  Main ingredients in the proof are the matrix-forest theorem for the relaxation behavior of the heat flux, and the matrix-tree theorem giving the low-temperature asymptotics of the stationary probability. The main new condition for the extended Third Law requires the absence of major (low-temperature induced) delays in the relaxation to the steady dissipative structure.
\end{abstract}
\maketitle

\section{Introduction}

The Nernst Postulate (1907) states that the change of entropy in any isothermal process approaches zero as the temperature reaches absolute zero\footnote{Its historical origin lies in the variational principle of Thomsen and Berthelot, which was an empirical precursor of the Gibbs variational principle, \cite{callen}.}. It evolved into the Planck version of the Third Law of Thermodynamics, stating that the entropy of a perfect crystal of a pure substance vanishes at absolute zero. That Third Law differs from the First and Second Laws, as it does not automatically follow from more microscopic considerations.  In fact, it sometimes fails, as exemplified by ice models in \cite{paul}. Nevertheless, theorems on the validity of the Third Law have been obtained for a large class of equilibria, cf. \cite{aL}:
for quantum and classical lattice systems, the entropy density at absolute zero temperature is directly related to the degeneracy of the ground state
corresponding to boundary conditions with the highest degeneracy.  The Third Law holds for example for any ferromagnetic Ising system.  Counterexamples include dimer systems and (spin) ice models, \cite{kast,temp,lieb,paul}.\\

It is natural to ask about thermal properties of nonequilibria as well. Then appears the question of low-temperature asymptotics and a possible formulation of an extended Third Law.
Note that the temperature (which goes to zero) need not be physically associated to the degrees of freedom of the system but rather to a thermal bath or large equilibrium environment weakly coupled to the system.   The main point  to remember is that steady nonequilibrium systems are open (and possibly small) and constantly dissipate heat into the (much larger) environment.  When parameters change, such as the temperature of that heat bath or the volume of the system, the heat flux may change.  In other words, when connecting two nonequilibrium conditions via a quasistatic transformation, an excess heat may flow.  The present paper rigorously defines that notion of excess heat, and we prove, under certain conditions, that this excess heat vanishes at zero ambient temperature.  We do that for irreducible (continuous time) Markov jump processes with finite state space.  To connect their mathematics with thermal properties, we require the interpretation of local detailed balance, which identifies the heat during a transition with the logarithmic ratio of forward to backward rates.\\

In Section \ref{ms} we start with the setup, definitions and a presentation of the main results.  In particular, we specify the Markov jump processes with the physical interpretation of heat and dissipated power in a thermal bath. The heat flux is parameterized by the inverse temperature $\beta$ of the bath and by $n$ other (external) real parameters $\alpha$, summarized in $\lambda = (\beta^{-1},\alpha)\in \bbR^{n+1}$. We show that when making a quasistatic transformation over a curve $\Gamma$ in that parameter space, the expected excess heat to the heat bath equals
\begin{equation}\label{geoi}
    Q(\Gamma) = \int_\Gamma \id \lambda\cdot D_\lambda\, \qquad D_\lambda =  \langle \nabla_\lambda V_\lambda\rangle_\lambda^s
    \end{equation}
    in terms of the so-called quasipotential $V_\lambda(x)$.  The notation $\langle\cdot\rangle_\lambda^s$ indicates an average in the stationary probability distribution $\rho_\lambda^s $ at parameter values $\lambda$.\\
Toward the end of Section \ref{ms} we  already informally discuss our main result, that $D_\lambda$ vanishes as the temperature of the thermal bath tends to absolute zero.  In fact, $D_\lambda = (-C_\alpha(\beta), D_\beta(\alpha))$ splits up in two components, depending on whether $\beta$ is kept constant (in $D_\beta(\alpha)$), or if $\alpha$ is kept constant (in $C_\alpha(\beta)$) during the quasistatic transformation.  Theorem \ref{t1} gives conditions for $D_\beta(\alpha)$ to vanish at absolute temperature.  The heat capacity
\begin{equation}\label{hccc}
C_\alpha(\beta) =  \beta^2 \langle \frac{\partial}{\partial \beta} V_\lambda\rangle_\lambda
\end{equation}
is the variation with temperature of the excess heat toward the system.  We show that $C_\alpha(\beta)$ tends to zero as well for inverse temperature $\beta\uparrow \infty$ [Theorem \ref{t2}].\\
The conditions for both Theorems are introduced at the end of Section \ref{ms}.  We already discuss  the static condition in Section \ref{drho}, to obtain $\nabla_\lambda \rho_\lambda^s \rightarrow 0$ when $\beta\uparrow \infty$.\\
Section \ref{graph} introduces the graph-theoretic elements needed for the proof of the boundedness of the quasipotential.  It can be stated in terms of the matrix-forest theorem.  The sufficient conditions are illustrated in Section \ref{exam} with simple examples, to show that they are on target.\\
The actual proofs of the main theorems start in Section \ref{proofs}.  It first concentrates on showing the geometric result \eqref{geoi}.   Section \ref{proofsbp} shows the boundedness of $V_\lambda$ in $\beta\uparrow \infty$.\\
A summary of results and arguments is given in Section \ref{conc}.\\
The Appendix makes the explicit links with the matrix-forest theorem.\\

We end this introduction by giving some additional context and background.  Steady state thermodynamics has been started in a number of papers like \cite{prig,oono,kom2}, where it attempts to remove the condition of close-to-equilibrium that was prominent in much of  irreversible thermodynamics \cite{dGM}.  In particular, the idea of excess heat as used in the present paper has been discussed in \cite{oono,kom2}.  Then, around 2011, nonequilibrium heat capacities were explicitly introduced and examples were discussed in \cite{epl,cejp,cal,jir}.  Other definitions of nonequilibrium heat capacity have been proposed in various papers, including~\cite{man,subas,dls}.\\
For more than a decade then, it remained very much an open question whether and when those heat capacities would vanish at absolute zero. Only with the present paper, a precise answer is given.\\
A physics presentation of the mathematical results contained here is found in \cite{nernst}.  In particular, we discuss there how our setup covers certain quantum features and indeed fits the Nernst Postulate as usually presented in thermochemistry. We also give there a heuristic argument to show that Third Law behavior follows when relaxation times do not exceed the dissipation time.   More illustrations and calculations of nonequilibrium heat capacities are found in \cite{simon,pritha,drazin}.  The present paper gives the rigorous version including all mathematical details concerning the quasistatic limit of the excess heat, and the graphical-expression of the quasipotential that leads to the proof of its boundedness.  That has not appeared elsewhere.

\section{Setup and main results}\label{ms}

\subsection{Markov jump process}\label{mpr}
Consider a simple, connected and finite graph $G = (\cal V(G), \cal E(G))$ with vertex set $\cal V(G) = \cal V$ and edge set $ \cal E(G) = \cal E$.  Vertices are written as $x,y,\dots$; edges are denoted by $e:=\{x,y\}$ when unoriented and $\overline{e}:=(x,y)$ is an oriented (or, directed) edge which starts in $x$ and ends in $y$.\\
We consider a Markov jump process $X_t = X_t^\lambda \in \cal V$ for times 
 $t\geq 0$ and with rates $k_\lambda(x,y)>0$ for the transition $x$ to $y$ when $\{x,y\} \in \cal E$, and otherwise the rates are zero.  The $\lambda$ refers to the dependence of the process on $n+1$ real parameters $\mathbf{\lambda} = (\beta^{-1},\alpha), \alpha = (\alpha_i, i=1,\dots,n)$ with $\beta\geq 0$ interpreted as the inverse temperature of a thermal bath, and other parameters $\alpha\in {\cal A}$ for some open set ${\cal A} \subset \mathbb{R}^n$.  The graph $G$ is fixed throughout, and does not depend on $\lambda$. We assume that all the transition rates $k_\lambda(x,y)$ are smooth in $\lambda$, which implies smoothness of all derived quantities.\\ 
  By irreducibility, there is a unique stationary probability distribution $\rho^s_\lambda(x)>0,\,  x\in \cal V$, solution of
the stationary Master Equations,
\begin{equation}\label{maeq}
\sum_y k_{\mathbf \lambda}(x,y)\rho^s_\mathbf{\lambda}(x) = \sum_y k_{\mathbf \lambda}(y,x)\rho^s_\mathbf{\lambda}(y)
\end{equation}
Stationary expectations with respect to $\rho^s_\lambda$ are written as $\langle f \rangle^s_\lambda = \sum_x f(x) \rho^s_\lambda(x)$ for functions $f(x), x\in \cal V$.\\
The backward generator $L_\lambda$ is the $|\cal V|\times |\cal V|$-matrix having elements $L_\lambda(x,y)=k_\lambda(x,y)$ and $L_\lambda(x,x)=-\sum_y k_\lambda(x,y)$.  We ignore the $\lambda-$dependence in the notations if no confusion can arise.\\

As physical orientation, we think of the vertices as states or configurations of an open system.  The transitions between states  are  possibly accompanied by exchanges of energy or particles with the environment.  The $\beta$ is the inverse temperature of the environment and the parameters $\alpha$ may appear in (interaction or self-) energies, or can quantify external parameters such as  spatial volume or boundary conditions. 

\subsection{Excess heat}
We assume that
\begin{equation}\label{q(z,z')}
\frac 1{\beta}\log\frac{k_\mathbf{\lambda}(x,y)}{k_\lambda(y,x)} = q_\alpha(x,y)
\end{equation}
does not depend on $\beta$.  The $q_\alpha(x,y)$ are obviously antisymmetric.  Following the physical condition of local detailed balance \cite{ldb}, $q_\alpha(x,y)$ is interpreted as the heat to the thermal bath (at inverse temperature $\beta$) in the transition $x$ to $ y$.\\  We then write
\begin{equation}\label{power}
    {\cal P}_\lambda(x) := \sum_y k_\lambda(x,y) q_\alpha(x,y).
\end{equation}
for the expected instantaneous power when in $x$.  By convexity $\langle \cal P_\lambda \rangle_\lambda^s \geq 0$, and $\langle \cal P_\lambda \rangle_\lambda^s$ is called the stationary dissipated power.
An important quantity will be the \emph{quasipotential $V_\lambda$}, a function on $\cal V$ defined as
\begin{eqnarray}\label{prv}
    V_\lambda(x) &:=& \int_0^{+\infty} \id t\,
[\langle {\cal P}_\lambda(X_t)\,|\,X_0=x\rangle_\lambda - \langle \cal P_\lambda \rangle^s_\lambda]\\
&=& \int_0^{+\infty} \id t\,e^{tL}[{\cal P}_\lambda - \langle \cal P_\lambda \rangle^s_\lambda] (x)\label{pla}
\end{eqnarray}
where in the last equality appears the semigroup $S(t) = e^{tL}$ for which $\langle g(X_t)\,|\,X_0=x\rangle_\lambda = S(t)g\,(x)$ for an arbitrary function $g$.
The quasipotential will appear in the quasistatic expression of excess heat, to which we turn next.\\

Given a smooth time-dependence $\lambda(t), 0\leq t\leq 1$, of the parameters, we call $\Gamma$ its image as a curve in parameter space $\bbR^+\times \cal A$. For  a quasistatic process, we write $\lambda^\ve$ and consider a protocol where the system evolves under $\lambda^\ve(t) :=\lambda(\ve t), 0\leq t\leq \ve^{-1}$ on $\Gamma$. The $\ve$ is the rate of change in the parameter  protocol where $\lambda$ evolves.  The limit where $\ve\downarrow 0$ will make the process to become quasistatic. For such a time-dependent process, at every moment $t$ the distribution is $\rho_t^\ve$, solving the time-dependent Master equation,
\begin{equation}\label{time-dependent master eq}
    \frac{\partial}{\partial t}\rho^\ve_t = L^\dagger_{\lambda(\ve t)}\rho_t^\ve
\end{equation}
for the corresponding forward generator $L_\lambda^\dagger$  (transpose of $L_\lambda$), $L_\lambda^\dagger\rho_\lambda^s=0$.\\
For a given protocol $\lambda^{\ve}$ at rate $\ve>0$ over times $t \in [0, \ve^{-1}]$, the expected excess heat towards the thermal bath at inverse temperature $\beta$ is defined by
\begin{equation}\label{qai}
 Q_\ve := \int_0^{1/\ve} \id t\,\sum_x \big({\cal P}_{\lambda(\ve t)}(x)\rho^\ve_t(x) - {\cal P}_{\lambda(\ve t)}(x)\rho^s_{\lambda(\ve t)}(x)\big).   
\end{equation}
The quasistatic asymptotics of \eqref{qai} yields the ``geometric'' expression
    \begin{proposition}\label{pr1}
            \begin{equation}\label{geo}
    \lim_{\ve\downarrow 0} Q_\ve = Q(\Gamma) = \int_\Gamma \id \lambda\cdot D_\lambda
    \end{equation}
    with
    \begin{equation}\label{qpo}
    D_\lambda = \langle \nabla_\lambda V_\lambda\rangle_\lambda^s
    \end{equation}
    in terms of the quasipotential \eqref{prv}.
    \end{proposition}
Note that there is no problem with the regularity of the quasipotential as function of $\lambda$.  In fact, from \eqref{pla}, $V_\lambda$ solves the Poisson equation
\begin{equation}\label{poi}
LV_\lambda(x)  = f\,,\qquad 
\langle V_\la\rangle^s_\lambda =0
\end{equation}
with source $f(x) = \langle {\cal P}_\lambda \rangle_\lambda^s - {\cal P}_\lambda(x)$ satisfying the centrality condition $\langle f\rangle^s_\lambda =0$.  Regularity of the solution of  such Poisson equations has been studied in many contexts; see e.g. \cite{poiss} \cite{pm}.\\

We call \eqref{geo} geometric because the integral is over the curve $\Gamma$ in parameter space.  We have
\begin{eqnarray}\label{cl}
D_\lambda &=& -\sum_x V_\lambda(x)\,\nabla_\lambda \rho_\lambda^s(x) =
-\big\langle V_\lambda\,\nabla_\lambda \log \rho^s_\lambda \big\rangle_\lambda^s
\nonumber\\
&=:&  \left(-C_\alpha(\beta),D_\beta(\alpha)\right)\label{cla}
    \end{eqnarray}
where the first component is minus the heat capacity, {\it cf}. \eqref{hccc}, 
\begin{equation}\label{hca}
C_\alpha(\beta) = -\Big\langle \frac{\partial}{\partial \beta^{-1}} V_\lambda \Big\rangle^s_\lambda = \big\langle V_\lambda\,\frac{\partial}{\partial \beta^{-1}} \log \rho^s_\lambda \big\rangle_\lambda^s
\end{equation}
which fixes the parameters $\alpha$ in the variation.  
The remaining $n-$component $D_\beta(\alpha)$ has  fixed (large) $\beta$ over the curve $\Gamma$ in \eqref{geo}:
\begin{equation}
 D_\beta(\alpha) =  -\big\langle V_\lambda\,\nabla_\alpha \log \rho^s_\lambda \big\rangle_\lambda^s  
\end{equation}
which is physically related to latent heat.\\
From \eqref{cla}, the strategy of finding sufficient conditions  for $D_\lambda\rightarrow 0$ as $\beta\uparrow \infty$, is clearly suggested:  we want to find conditions so that, (1) that for all $x\in \cal V$,
\begin{equation}\label{stac}
\lim_{\beta\uparrow\infty}\nabla_\lambda \rho_\lambda^s(x) =0
\end{equation}
and (2) that the quasipotential is uniformly bounded in $\beta$: there is a constant $c_v$ so that for all $x$,
\begin{equation}\label{dynaco}
|V_\lambda(x)| < c_v, \,\quad\alpha \in {\cal A}
\end{equation}
We start with the first condition \eqref{stac} in the next section, involving the low-temperature asymptotics of $\rho_\lambda^s$.  The study of the second condition starts with introducing graph elements in  Section \ref{graph} and proving the boundedness in Section \ref{proofsbp} with Proposition \ref{bdd}.  
However, before continuing, it is interesting to check with the equilibrium situation.  Equilibrium dynamics is a detailed balance dynamics, characterized by the existence of a potential function $V_\alpha$ on $\cal V$ such that $q_\alpha(x, y) = V_\alpha(x) - V_\alpha(y)$.  Then, the quasipotential is actually equal to $V_\alpha(x) - \langle V_\alpha\rangle^s_\lambda$ and the excess heat reduces to the standard heat. Note then that  in the equilibrium case, only  condition \eqref{stac} is needed, and the
extended Third Law becomes the standard Third Law. The boundedness \eqref{dynaco} only enters when the system is out of equilibrium, but it is essential (as can be seen from examples) even arbitrarily
close to equilibrium.

\subsection{Low-temperature stationary measure}\label{drho}
Since we are dealing with systems having a finite number of states, and no estimates on the behavior of the excess heat as function of the number of particles are attempted, we only need to worry for \eqref{stac} about the low-temperature asymptotics of the stationary probability.  We recall first the general result on the low-temperature structure of stationary measures for Markov jump processes; see \cite{heatb,lowT}.\\

The Kirchhoff formula for the stationary distribution  $\rho^{s}_\lambda$ reads,
\begin{eqnarray}\label{kir}
    \rho^s(x) &=&\frac{w(x)}{W}, \quad w(x)=\sum_{ T\in \cal T} w( T_x)\\
    w( T_x)&=& \prod_{\overline{e}=(z,z')\in  T_x} k(z,z')\label{kir2}
\end{eqnarray}
where  $W$ is the normalization and the sum in the weights $w(x)$ is over all spanning trees $ T$ in the graph, \cite{kir}, take $\cal T$ as the set of all spanning trees in graph $G$.  For a given spanning tree $T$, $ T_x$ is its oriented version with all edges directed to $x$.  The weight $w(T_x)$ is the product of transition rates over the oriented edges $\overline{e}=(z,z')$ in that oriented spanning tree with root $x$.\\ 
We introduce
\begin{equation}
 \phi_\alpha(x,y):=\lim_{\beta \rightarrow \infty }
   \frac{1}{\beta}\log  k_\lambda(x,y)   
\end{equation}
and for an oriented subgraph $\cal G$ of $G$, write  $\phi_\alpha(\cal G):=\sum_{(z,z')\in \cal E(\cal G)}\phi_\alpha(z,z')$.  Define
\begin{align}\label{phi}
    \phi_\alpha(z):=\max _{T\in \cal T} \phi_\alpha(T_z), \quad \phi^*_\alpha:=\max_{z\in \cal V} \phi_\alpha(z)
\end{align}
We call $x$ a dominant state whenever $\phi_\alpha(x) = \phi^*_\alpha$.\\
As shown in \cite{heatb}, formul{\ae} \eqref{kir}--\eqref{kir2} lead to the low-temperature asymptotics, 
\begin{equation}\label{lowkir}
    \rho^s_\lambda(x) = B_\lambda(x)\,e^{-\beta [\phi^*_\alpha - \phi_\alpha(x)]}( 1+ O(e^{-\delta \beta}))
\end{equation} 
where $B_\lambda(x)$ is subexponential in $\beta$, i.e.\ $\log B_\lambda(x) = o(\beta)$, and
$\delta>0$, uniformly in  $\alpha\in \cal A$. See also \cite{lowT,intr}.\\

The asymptotics \eqref{lowkir} implies that two dominant states may still have different probabilities in the $\beta\uparrow\infty$ -limit, determined by the prefactor $B_\lambda(x)$.\\

Recall that the parameters $\alpha$ are unrelated to temperature; see Section \ref{mpr}.  In the conditions below, it is important that the zero-temperature behavior is obtained uniformly in $\alpha \in \cal A$.   \\

\textbf{Condition 1a}: There is a 
probability distribution $\rho^0$ independent of $\alpha\in \cal A$ so that 
\begin{equation}\label{c12}
\lim_{\beta\uparrow\infty} \rho^s_\lambda(x) = \rho^{0}(x)
\end{equation}
uniformly in $\alpha$.\\

Condition 1a is implied when the maximizer of $\phi^*_\alpha$ in equation \eqref{phi} is unique but in \eqref{c12} we do not require that the correction is exponentially small in large $\beta$.
\begin{proposition}\label{condition1.2}
    Assume   Condition 1a.  Then,  $\lim_\beta \nabla_\alpha \rho^s_\lambda(x) = 0$ for all $x\in \cal V$.
\end{proposition}
\begin{proof}
By the assumed smoothness of the transition rates in $\alpha$ for every $\beta$ and since the number of vertices is finite, we have that the stationary distribution $\rho^s_{\lambda}$ is smooth as well.  
By the uniform limit, we can exchange the limits ``$\lim_\beta \nabla_\alpha= \nabla_\alpha \lim_\beta$'' and we have the required $\nabla_\alpha \rho^s_\lambda(x) \rightarrow 0$ as $\beta\uparrow\infty$.
\end{proof}
In the following condition, we no longer need that the zero-temperature limit of the stationary distribution is $\alpha-$independent, but the convergence speed must be faster than $\beta^{-1}\,$, uniformly in $\alpha\in \cal A$.\\

\textbf{Condition 1b}: There is a probability distribution $\rho_\alpha^o$ so that
\begin{equation}\label{c2}
\lim_{\beta\uparrow\infty} \beta\left(\rho^s_\lambda(x) - \rho^o_\alpha(x)\right) = 0
\end{equation}
uniformly in $\alpha\in \cal A$, for all $x \in \cal V$.\\

As a direct consequence:
\begin{proposition}\label{condition2}
 Condition 1b implies that $\lim_\beta \frac{\partial \rho^s_\lambda(x)}{\partial \beta^{-1}}  = 0$ for all $x\in \cal V$.
\end{proposition}
We can summarize the situation so far as follows. Recall the Definitions \eqref{qpo} and \eqref{hca}, of excess heat and heat capacity.
\begin{proposition}
    When \eqref{dynaco} (the quasipotential is uniformly bounded for $\beta\uparrow\infty$), and Condition 1a are satisfied, then  $D_\beta(\alpha)\rightarrow 0$ as $\beta\uparrow \infty$. When \eqref{dynaco} holds together with Condition 1b, then the heat capacity $C_\alpha(\beta) \rightarrow 0$ as $\beta\uparrow \infty$.
\end{proposition}
For the rest of the paper we can thus concentrate on the boundedness of $V_\lambda$. Proofs are collected in Section \ref{proofs}. It is important here to understand the role of graphical representations as in the next section.  We give some intuition to motivate that approach:\\
In the illustrations of heat capacities in the literature so far, it was already observed how their behavior as function of temperature may inform us about dynamical aspects as well.  The main change from equilibrium is indeed that nonequilibrium heat capacities are able to pick up the dependence of the relaxation of heat and dissipated power on temperature and other parameters.  The low-temperature behavior of heat capacities is therefore not only or no longer only informing us about static fluctuations but information about {\it dynamical accessibility} enters as well, which is naturally expressed as a weighted graph-property. That is seen most spectacularly in the behavior of heat capacities in the immediate neighborhood of absolute zero.  As we will prove and illustrate, all states must remain ``relatively well-connected,'' which is a dynamical condition. Even when the invariant measure at zero temperature is concentrating on one unique state, the heat capacity can still be different from zero (and even diverge).  That happens at parameter values where the relaxation behavior of the heat gets pathological; we can speak here of a localization-phenomenon where there is a local delay in the relaxation to the stationary dissipation.   Such ``accessibility'' and ``no-delay'' can be captured more precisely by the graphical elements that we next introduce, and will be illustrated via examples also in Appendix \ref{exam}.  The dynamical Condition 2 in the beginning of Section \ref{mair} will be summarizing all we need.

\section{Graph elements}\label{graph}
We recall some standard notions from graph theory \cite{bala}.  Remember that we have  a connected and finite simple  graph $G = (\cal V, \cal E)$ with vertices $x,y,\dots$ and edges denoted by $e:=\{x,y\}$ when unoriented and $\overline{e}:=(x,y)$ if oriented (or, directed) from $x$ to $y$.\\
A graph $H$ is called a \textit{subgraph} of  $G$ if $\cal V(H) \subseteq \cal V(G)$ and $\cal E(H) \subseteq \cal E(G)$. We call $H$ a {\it spanning} subgraph if $\cal V(H) = \cal V(G)$ and  $\cal E(H) \subseteq \cal E(G)$. If every edge of $H$ is assigned a direction  then  an oriented subgraph of $G$  is created and it  is  shown by $\overline{H}$.\\
A \textit{path} in  graph $G$ is an alternating sequence of vertices ($x_i$) and directed edges ($\overline{e}_j$): $x_0\,\overline{e}_1\,x_1, \overline{e}_2\, x_2...\overline{e}_n\,x_n$ in which edge $\overline{e}_i$ starting from vertex $x_{i-1}$ and ending to vertex $x_i$ and also  all visited vertices and edges are different. For example, $(x,(x,y),y,(y,z),z)$ is a path from $x$ to $z$.\\
A \emph{loop} denoted by $\ell$ is an alternating sequence of vertices ($x_i$) and edges ($e_j$): $x_0\,e_1\,x_1, e_2\, x_2...e_n\,x_0$ in which edge $e_i$ ends in $x_{i-1}$ and  $x_i$. In the sequence   the initial and the final vertex are the same and the other vertices and edges are all different.\\
An \textit{oriented loop} denoted by $\overline{\ell}$ is an alternating sequence of vertices ($x_i$) and directed edges ($\overline{e}_j$): $x_0\,\overline{e}_1\,x_1,\overline{e}_2\, x_2...\overline{e}_n\,x_0$ in which edge $\overline{e}_i=(x_{i-1},x_i)$. In the sequence   the initial and the final vertex are the same and the other vertices and edges are all different. For example, $(x,(x,y),y,(y,z),z,(z,x),x)$ is an oriented loop. \\

A connected simple graph without a loop is called  a \textit{tree}; a single vertex being also considered as  a tree.  A \textit{rooted tree} is an oriented tree such that all edges are directed toward a specific vertex  called \emph{root}. In  a tree  rooted in $x$, there is always a unique path from any other vertex to $x$ and there is no edge going out from $x$. A spanning subgraph of $G$ without any loop is called a \emph{spanning tree} in $G$, and when rooted at some $x$ then it is called a \emph{rooted spanning tree}. The set of all spanning trees in $G$ is denoted by $\cal T$. For more clarity, we reserve the symbol $T$ (or $T_x$) for spanning trees (or rooted spanning trees) in $G$, and the symbol $\tau$ (or $\tau_x$) to denote trees (or rooted trees) which are not spanning.  See example~\ref{tree} and Fig.~\ref{ext}. 

\begin{definition}\label{tl}
    \begin{enumerate}[label=(\alph*)]
 \item A \emph{tree-loop} is a graph made by a loop and  trees connected to the loop; see Fig.~\ref{treecycle}.  Clearly, a loop is considered to be a tree-loop in which every vertex on the loop is considered to be a tree.  However, a tree is not considered to be a (special) tree-loop.
         \begin{figure}[H]
    \centering
    \includegraphics[scale=0.35]{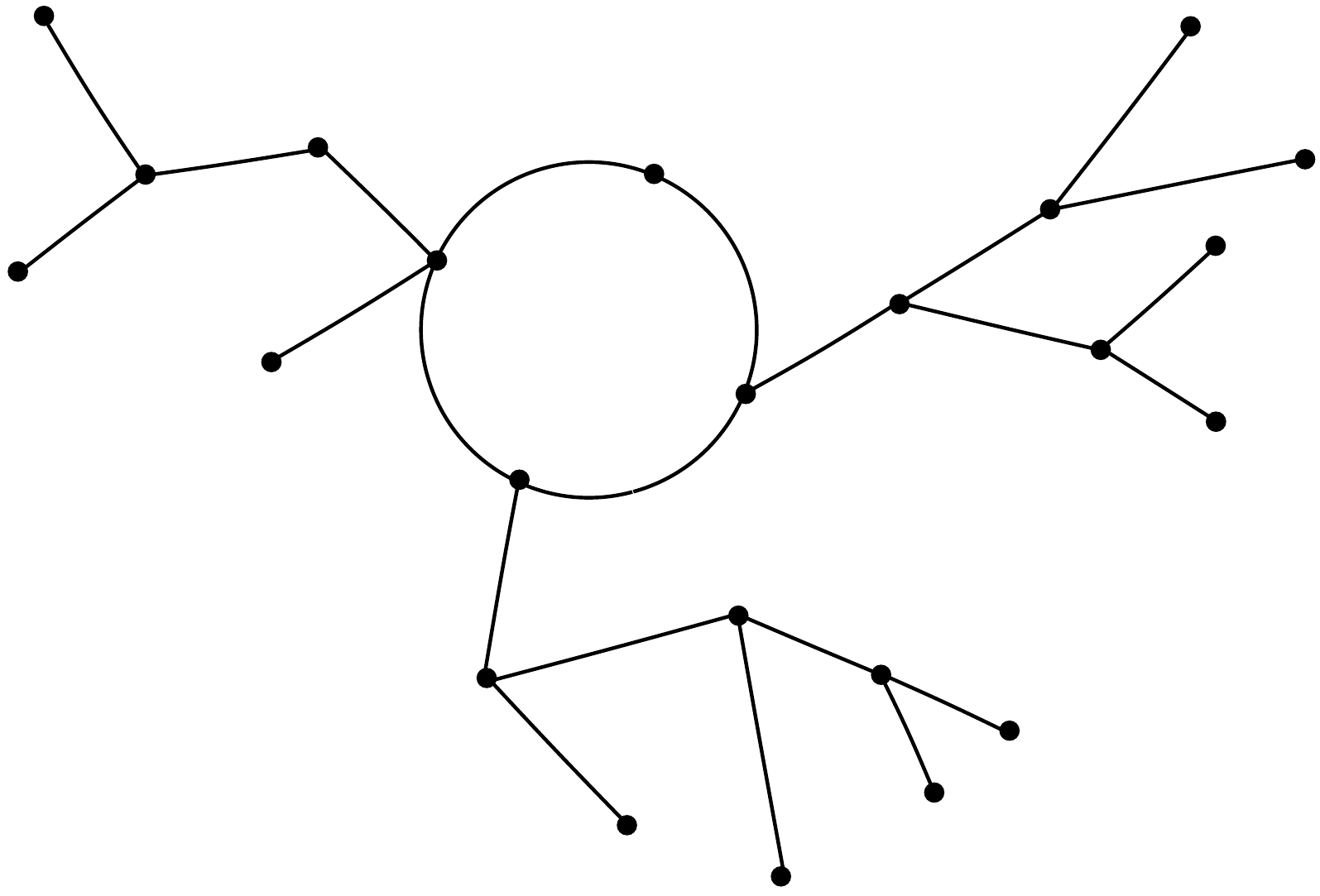}
   \caption{A tree-loop with four trees, one of which is a single vertex.}
    \label{treecycle}
    \end{figure}
 An \emph{oriented tree-loop} is a graph consisting of an oriented loop and rooted trees such that the root of every tree is located on the loop; see Fig.~\ref{otl}.
        \begin{figure}[H]
     \centering
     \begin{subfigure}{0.49\textwidth}
         \centering
         \includegraphics[scale = 0.35]{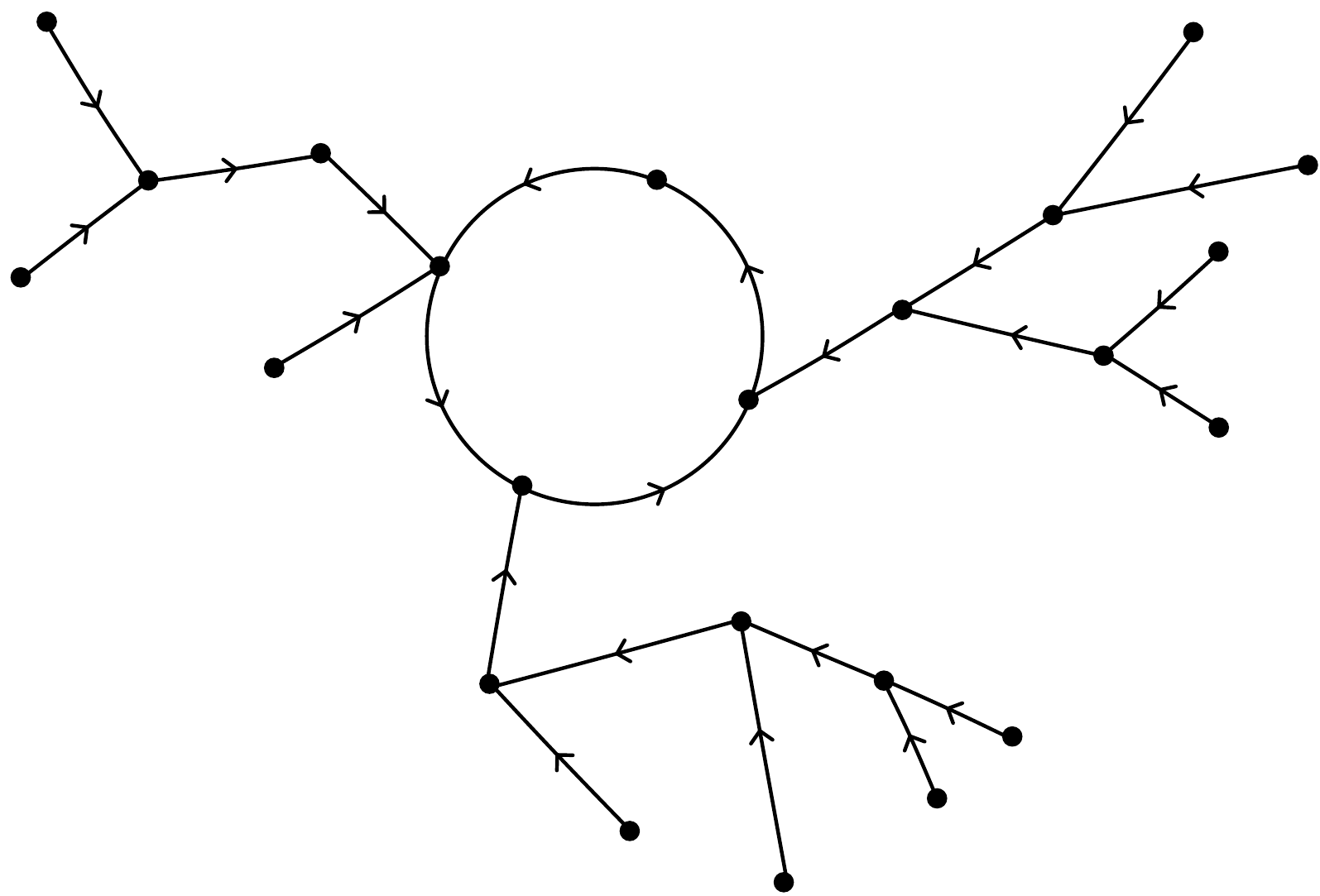}
         \caption{counter clockwise tree-loop}
         \label{fig:counter}
     \end{subfigure}
     \hfill
     \begin{subfigure}{0.49\textwidth}
         \centering
         \includegraphics[scale = 0.35]{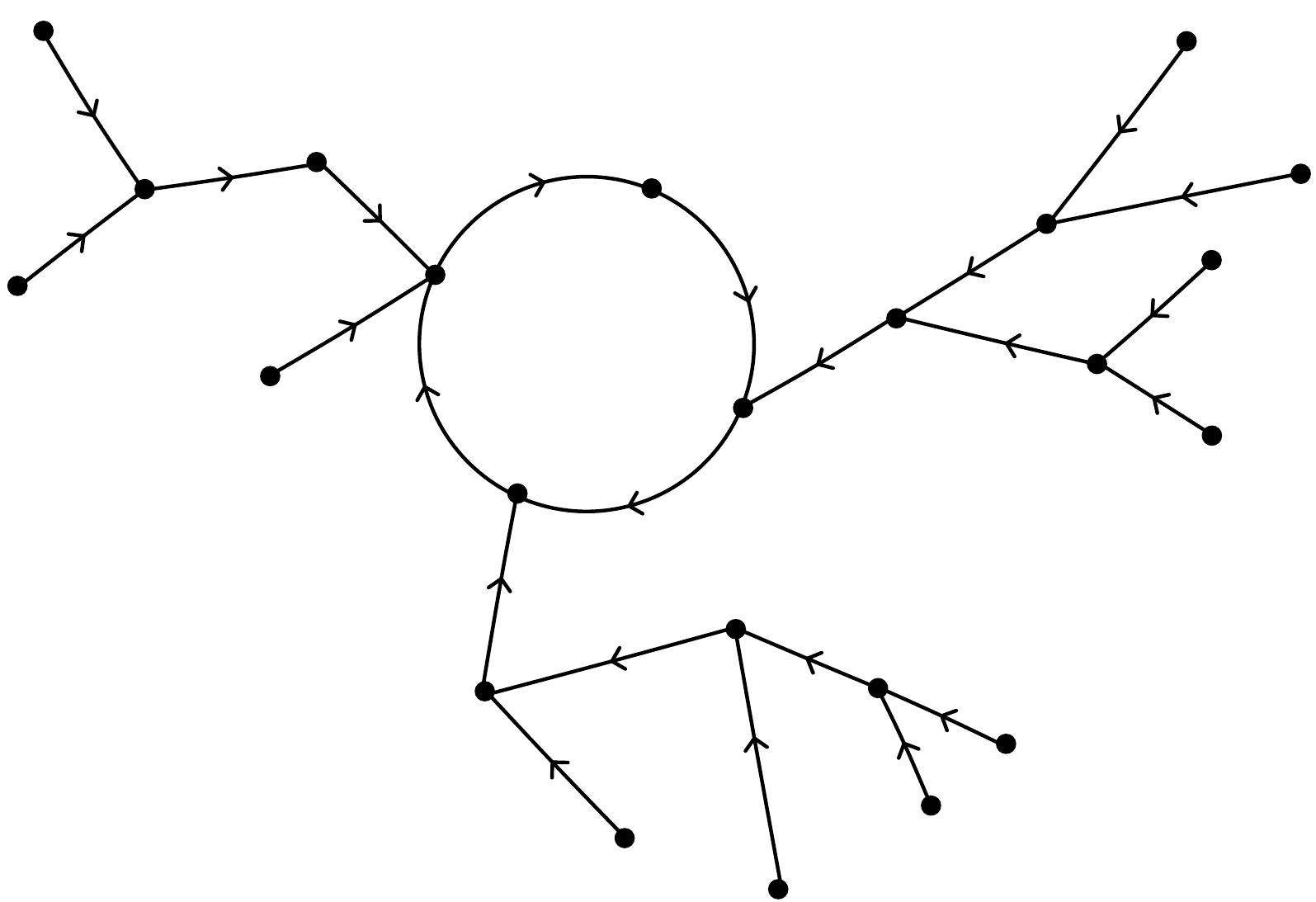}
         \caption{clockwise tree-loop}
         \label{fig:clock}
     \end{subfigure}
        \caption{\small{Two oriented tree-loops which only differ in the orientation on the loop. Remark that in every tree-loop the roots of the trees are located on the loop}.}\label{otl}
        \label{fig:tree}
\end{figure}
\item A \emph{spanning tree-loop} in graph $G$ is a  spanning subgraph of $G$ which is a tree-loop. An  \emph{oriented spanning tree-loop} in graph $G$ is a spanning tree-loop of $G$ which is oriented such that trees are rooted in the loop and the loop is directed in  either clockwise or counter clockwise direction.  We denote the set of all spanning tree-loops in the graph $G$ by $\cal H$. For a specific edge $e\in \cal E(G)$,   $\cal H_e$ is  set of all spanning tree-loops of $G$ such that the edge $e$ is part of a tree. Put $O(\cal H_e)$ as  the set of all oriented tree-loops including the edge $e$ located on a tree. See Example \ref{not}, and more examples in Appendix \ref{exam}.
    \end{enumerate}
\end{definition}

\begin{definition}\label{tlt}
    \begin{enumerate}[label=(\alph*)]
 \item A \emph{tree-loop-tree} is a graph consisting of two disconnected simple graphs, one of which is a  tree and the other one is a tree-loop.  We can make it by removing an edge from the tree-parts in a tree-loop; see Fig.~\ref{otct}. 

         \begin{figure}[H]
            \centering
            \includegraphics[scale=0.35]{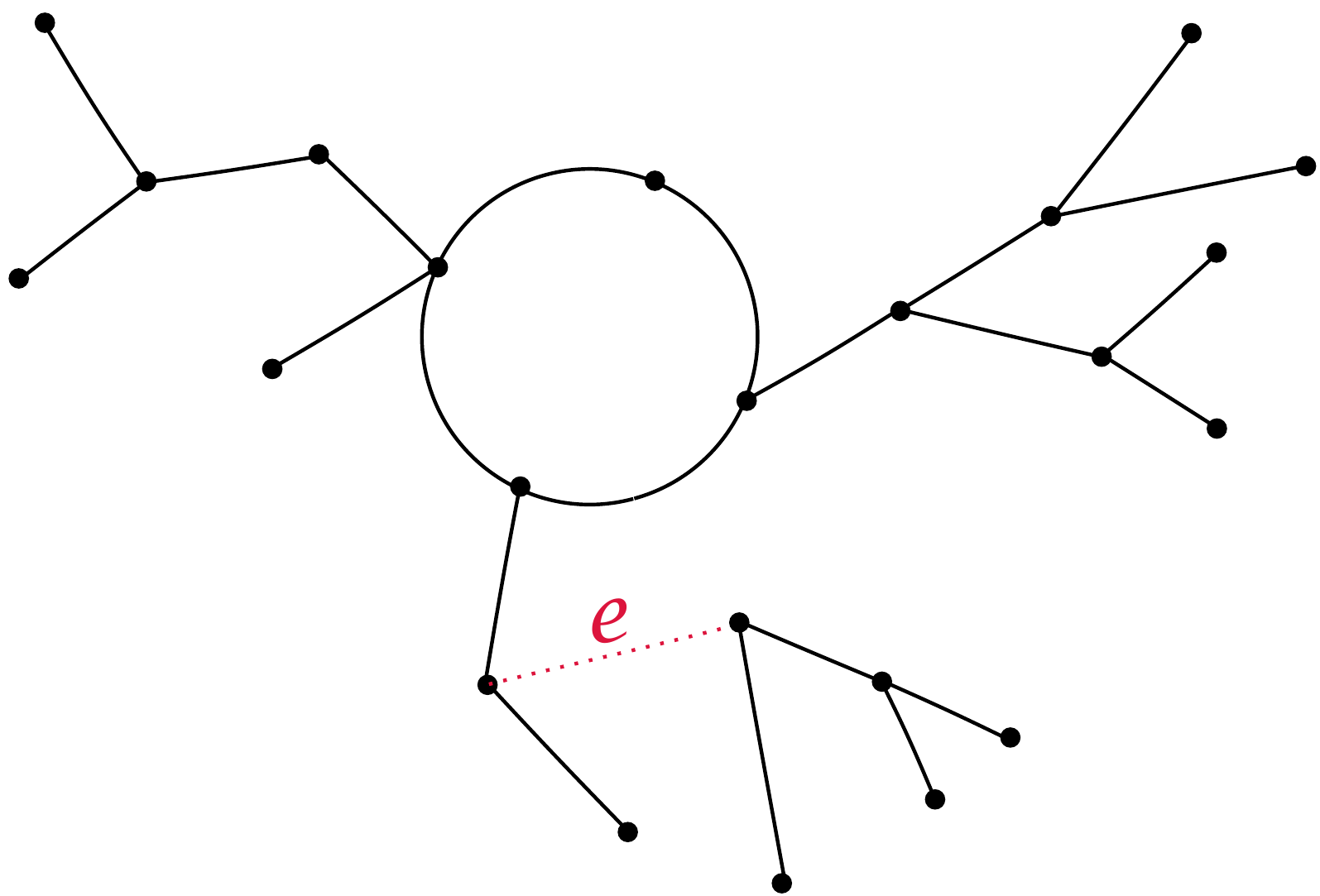}
            \caption{\small{A  tree-loop-tree made by removing the edge $e$ from a tree-loop.}}\label{otct}
        \end{figure}
        
 An \emph{oriented tree-loop-tree} is a tree-loop-tree where every edge has a direction. The directions of the edges are such that the tree-loop part is an oriented tree-loop and the tree part is a rooted tree; see Fig.~\ref{otct+}.  In the present paper every (oriented) tree-loop-tree will be a \emph{spanning} (oriented) tree-loop-tree, defined as follows.
 \begin{figure}[H]
            \centering
            \includegraphics[scale=0.35]{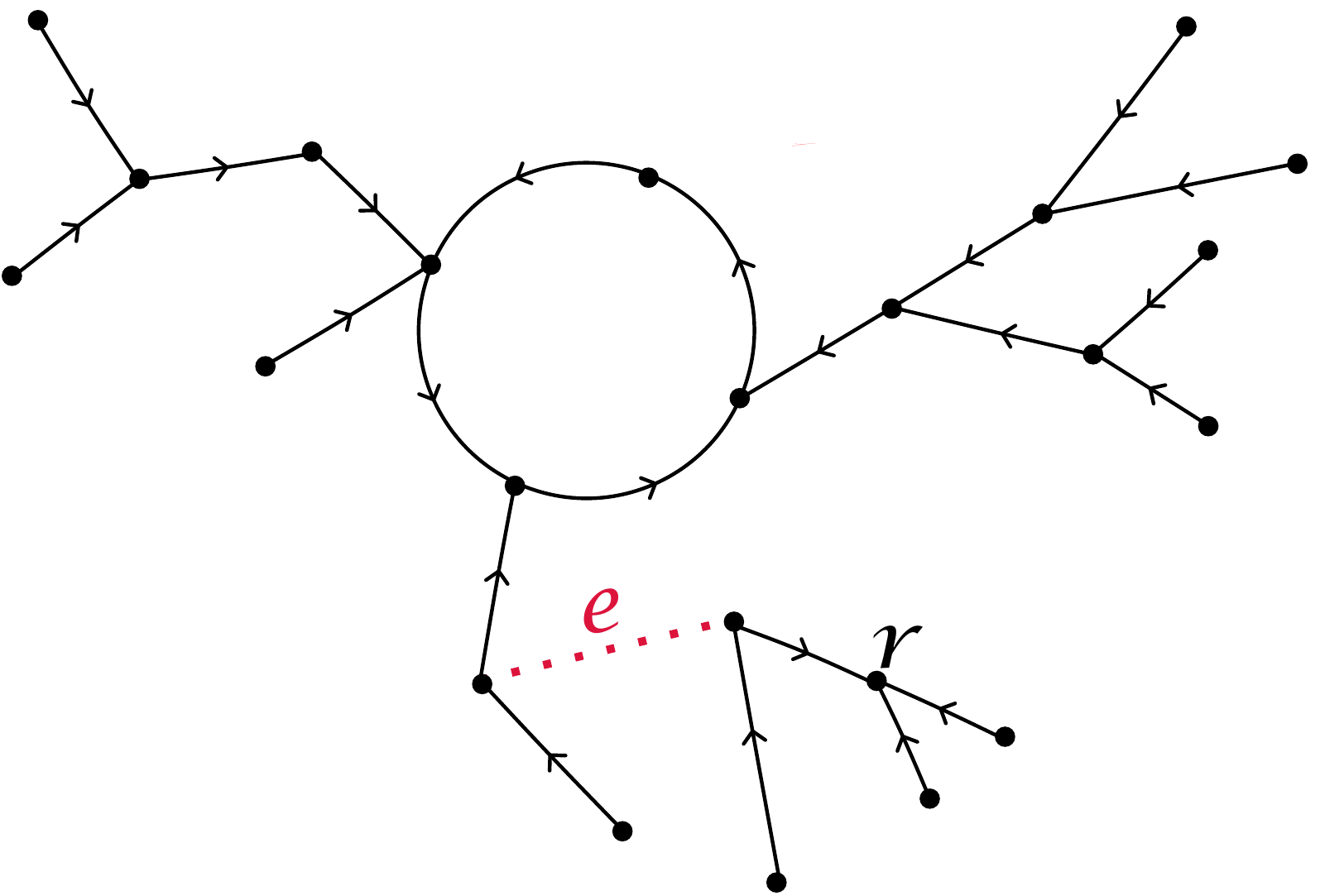}
             \caption{\small{An oriented tree-loop-tree made by removing the edge $e$  from a tree-loop and adding orientations. The orientation in the loop is counter-clockwise and the tree part is rooted in $r$.}}\label{otct+}
        \end{figure} 
 \item A \emph{spanning tree-loop-tree} in the graph $G$ is a spanning subgraph of $G$, which is a tree-loop-tree. An \emph{oriented spanning tree-loop-tree} in graph $G$ is a spanning tree-loop-tree in the graph $G$ such that the tree part is a rooted tree and the tree-loop part is an  oriented tree-loop.  The set of all spanning tree-loop-trees in the graph  $G$, which are  made by removing the edge $e$ from the tree-loops in the set  $\cal H_e$ is denoted by $\cal H^{(e)}$.  $O(\cal H^{(e)})$ is the set of all oriented tree-loop-trees from  $\cal H^{(e)}$. Notice that every tree-loop-tree from $\cal H^{(e)}$ has different possible orientations. If $H\in \cal H_e$ is a spanning tree-loop, then $O( H^{(e)})$ denotes the set of all oriented spanning tree-loop-trees which are made by removing the edge $e$ of a specific tree-loop (denoted by $H$) and then giving it all possible orientations  to the resulting tree-loop-tree; see Example \ref{not}. (There are two different orientations possible in the loop and different orientations are possible in the separate tree depending on the location of the root in that tree.) In summary, $H^{(e)} \in \cal H^{(e)}$ and $O(H^{(e)})\subseteq O(\cal H^{(e)})$. We define $\sigma_H(\overline{e}) = -1$ if $\overline{e}$ is oriented towards the loop of $H$ and $\sigma_H(\overline{e}) = +1$ otherwise.
        \end{enumerate}
\end{definition}

We define the weight  $w(\overline{e})$ for an oriented edge as $w(\overline{e})=k(\overline{e})$ . The weight of an oriented subgraph $\overline{H}$ with $\cal E(\overline{H}) \neq \emptyset$ is 
\begin{equation}\label{ww}
w(\overline{H})= \prod_{\overline{e} \in \cal E(\overline{H})}w(\overline{e}).
\end{equation}
If there is no edge in $\overline{H}$ but it still has vertices, then we define $w(\overline{H})=1$. The weight of the empty set is $w(\emptyset) = 0$. We define $W=\sum_{x \in \cal V} \sum_{T\in \cal T} w(T_x)$ as the sum over the weights of all rooted spanning  trees in the given graph. 
\begin{example}\label{not}
We illustrate  Definitions \ref{tl} and \ref {tlt}. Consider the graph in Fig.~\ref{exgraph}. 
\end{example}
\begin{figure}[H]
	\includegraphics[scale=0.5]{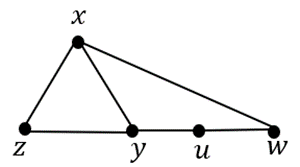}
	\caption{\small{Model graph for Example~\ref{not}}}
 \label{exgraph}
\end{figure}
Put
\begin{equation}e_1:=\{x,w\},\, e_2:=\{w,u\},\,e_3:=\{u,y\}\, e_4:=\{y,z\},\, e_5:=\{z,x\},\, e_6:=\{x,y\}.\end{equation}
The graph  has six possible spanning tree-loops, shown in Fig.~\ref{extl}. 
\begin{figure}[H]
	\includegraphics[scale=0.5]{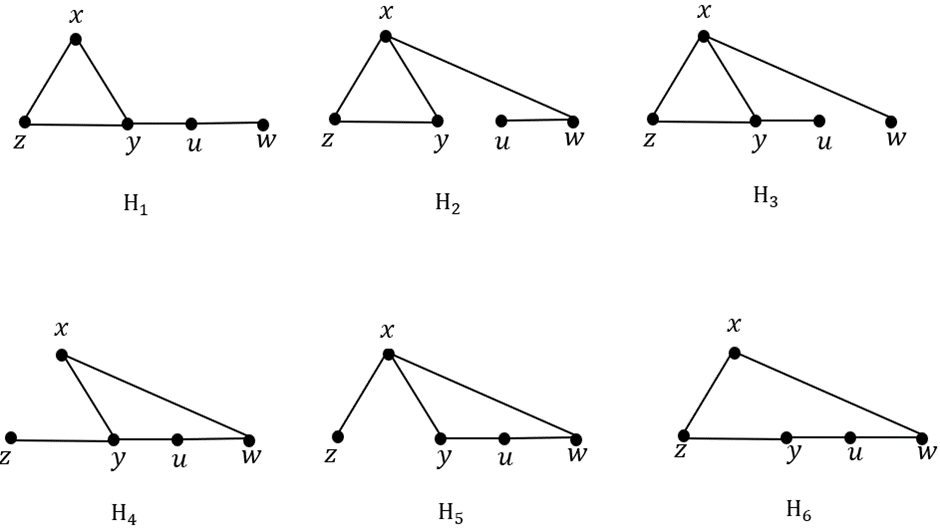}
	\caption{\small{All tree-loops of the graph in Fig.~\ref{exgraph}}} \label{extl}
\end{figure}
 For every edge, a set of spanning tree-loops  exists such that all components include that edge on the tree,
 \begin{align*}
    \cal H_{e_1}=\{H_2,H_3\},\, \,\cal H_{e_2}=\{H_1,H_2\},\, \,\cal H_{e_3}=\{H_1,H_3\},\, \, \cal H_{e_4}=\{H_4\},\, \, \cal H_{e_5}=\{H_5\},\, \,\cal H_{e_6}=\emptyset.
 \end{align*}
In Fig.~\ref{oextl} we show the set $O(H_1^{(e_3)})$ created by removing the edge $e_3$ from the tree-loop $H_1$.
\begin{figure}[H]
	\includegraphics[scale=0.45]{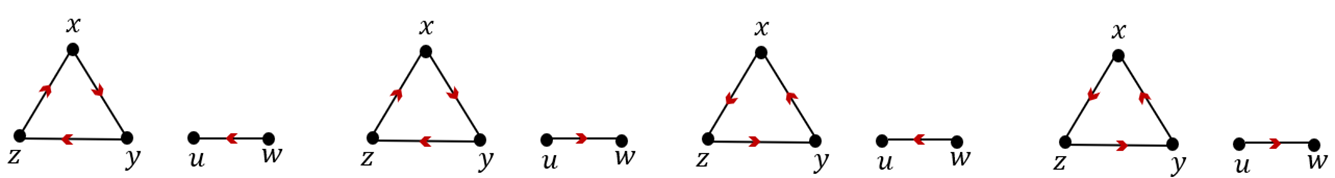}
	\caption{\small{The components of the set $O(H_1^{(e_3)})$.}} \label{oextl}
\end{figure}

\section{Main results}\label{mair}
To state the main results, 
we need\\ 
\textbf{Condition 2}: For every edge $e$ in the graph $G$ and for every element $\overline{H}$ of the set $O(\cal H^{(e)})$ (see Definition.~\ref{tlt}(b))
 \begin{equation}\label{c11}
\phi( \overline{H}) \leq \phi^* .
 \end{equation}
 That condition will be interpreted as a `no-delay' condition; see Example \ref{nod}.\\

The main result is a Nernst heat theorem for the given context.  Recall Proposition \ref{pr1}.
\begin{theorem}\label{t1}
Recall definition \eqref{qpo}. Under Conditions 1a and 2,
\begin{equation}\label{mr1}
\lim_{\beta\uparrow\infty} D_\beta(\alpha) = 0
\end{equation}
\end{theorem}
In Appendix \ref{exam}, it is illustrated how Condition 2 is natural.\\
We next consider the heat capacity \eqref{hca} in the limit $\beta\uparrow \infty$.
\begin{theorem}\label{t2}
Recall definition \eqref{hca}.
Under Conditions 1b and 2,
\begin{equation}\label{mr2}
\lim_{\beta\uparrow\infty} C_\alpha(\beta) = 0. 
\end{equation}
\end{theorem}

We end the section by stating the key to the main results, which is a graphical representation of the difference-quasipotential over any edge $\overline{e}=(x,y)$.  
\begin{proposition}\label{dec}
The difference-quasipotential $V(\overline{e}) = V_\lambda(x) - V_\lambda(y)$ (from \eqref{prv}) can be decomposed as
$V(\overline{e}) = V_{\text{tree}}(\overline{e}) +V_{\text{loop}}(\overline{e})$, with
\begin{equation}\label{treevedge}
V_{\text {tree}}(\overline{e})=\dfrac{1}{W}\; \sum_{T \in \cal T} \,q_{T}(x\rightarrow y ) \sum_{u\in \cal V} w(T_u)
\end{equation}
where the first sum is over all spanning trees, $q_{T}(x\rightarrow y ) $ is the sum of $q(z,z')$ (defined in \eqref{q(z,z')}) over the path from $x$ to $y$ located on the tree $T$,
and
\begin{align}\label{loopvedge}
V_{\text{loop}}(\overline{e})=\frac{1}{W}\,\sum_{H\in \cal H_{e}}\,\sigma_H(\overline{e})\, \sum_{\overline{H}\in O(H^{(e)})} q(\overline{\ell})\,w(\overline{H})
\end{align}
 in which the oriented loop $\overline{\ell}$ is  the one in $\overline{H}$ and   $\sigma_H(\overline{e})=\pm 1$ depending on whether $\overline{e}$ looks away or towards the loop (see Definition \ref{tlt}(b)).  
\end{proposition}

\underline{Remark}: There are other representations of $V_{\text{tree}}(\overline{e})$ such as from observing that
\begin{equation}
V_{\text{tree}}(\overline{e})=q(\overline{e})+\frac{1}{W}\,\sum_{T\in \cal T} \,\big(q_{T}(x\rightarrow y )+q(y,x)\big) \sum_{u\in \cal V} w(T_u)
\end{equation}
where the second term makes a loop. That is useful for interpretation as it reduces to $V(\overline{e}) = V_{\text{tree}}(\overline{e})=q(\overline{e})$ in the case of equilibrium dynamics, i.e.~such that the sum of heat 
$q(z,z')$ along edges of any loop vanishes (the ``global'' detailed balance). More generally, only heat-carrying loops contribute.

\section{Proofs}\label{proofs}

\subsection{Quasistatic analysis}\label{qa}
Note that always $\sum_x L^\dagger\rho (x) = \sum_x L1(x)\,\rho (x)= 0$ because $L1(x)=0$.  Hence, when $L^\dagger\rho = \nu$, then $\sum_x \nu(x) = 0$ must be true necessarily.  It means that the (nonexisting) inverse $1/L^\dagger$ can only be defined, if at all, on those $\nu$ such that $\sum_x\nu(x) =0$.  It implies
\begin{equation}
\sum_x \int_0^\infty \id t\,e^{tL}[f(x)-\langle f\rangle^s]\,\nu(x) 
 = \sum_x \int_0^\infty \id t\,e^{tL}f(x)\,\nu(x). 
\end{equation}
Therefore, we can actually define $(1/L^\dagger)\,\nu$  via the requirement that for all functions $f$,
\begin{equation}
\sum_x f(x) \,(\frac 1{L^\dagger}\nu) (x) = 
-\sum_x \int_0^\infty \id t\,e^{tL}[f(x)-\langle f\rangle^s]\,\nu(x) 
\end{equation}
or
\begin{equation}
\frac 1{L^\dagger}\nu(x) = -\sum_y \int_0^\infty \id t\,\big[\text{Prob}[X_t=x|X_0=y]-\rho^s(x)\big]\,\nu(y)
\end{equation}
which is indeed well-defined (and we do not bother introducing new notation for that restricted inverse).\\

We want to solve the time-dependent master equation perturbatively,
\begin{equation}
\rho^\ve_t = \rho^s_{\lambda(\ve t)} + \ve\,F_t^{(\ve)} + O(\ve^2)
\end{equation}
as everything is smooth around $\ve=0$.  Since time gets rescaled by $\ve$ we know automatically that $\frac{\partial}{\partial t}\,F_t^{(\ve)} = O(\ve)$.  On the other hand, 
\begin{equation}
\frac{\partial}{\partial t}\,F_t^{(\ve)} = \frac{\partial}{\partial t}\,[\rho^\ve_t - \rho^s_{\lambda(\ve t)}] = L^\dagger_{\lambda(\ve t)}F_t^{(\ve)} - \ve\,\dot{\lambda}(\ve t)\cdot \nabla_\lambda \rho^s_{\lambda(\ve t)}.
\end{equation}
Hence, 
\begin{equation}
F_t^{(\ve)} = \dot{\lambda}(\ve t)\cdot\frac 1{ L^\dagger_{\lambda}} \nabla_\lambda\rho^s_\lambda. 
\end{equation}
As a consequence,
\begin{equation}\label{mm}
\int_0^{ 1/\ve}(\rho^\ve_{t} - \rho^s_{\lambda(\ve t)})\id t = \int_\Gamma\id \lambda\cdot \frac 1{L_\lambda^\dagger}\nabla_\lambda\rho^s + O(\ve)
\end{equation}
where the integral on the right-hand side is purely geometrical, invariant under a reparametrization of time $t\mapsto \gamma(t)$ for a smooth function $\gamma$ with $\gamma'(t)>0$.\\

Suppose we have a function $f_\lambda$ on $\cal V$. We can always write $f_\lambda = -L_\lambda V_\lambda + \langle f_\lambda\rangle^s_\lambda$.  Define the excess
$f_\lambda^\text{exc} := f_\lambda - \langle f_\lambda\rangle^s_\lambda$, with time-integral
\begin{equation}
I_\ve := \int_0^{1/\ve} f^\text{exc}_{\lambda(\ve t)}(x_t)\, \id t .
\end{equation}
In the case of $f_\lambda = {\cal P}_\lambda$, from \eqref{qai} we have $Q_\ve = \langle I_\ve\rangle^\ve$.

    \begin{proof}[Proof of Proposition \ref{pr1}]
    From \eqref{mm},
       \begin{eqnarray}
\lim_{\ve\downarrow 0}\langle I_\ve\rangle^\ve &=& \lim_{\ve\downarrow 0}\int_0^{1/\ve} \id t\,\sum_x \big(f_{\lambda(\ve t)}(x)\rho^\ve_t(x) - f_{\lambda(\ve t)}(x)\rho^s_{\lambda(\ve t)}(x)\big)\nonumber\\
&=& \int\id \lambda\cdot\sum_x f_\lambda(x)\, \frac 1{L_\lambda^\dagger}\nabla_\lambda \rho^s_\lambda(x)\nonumber\\
&=& -\int\id \lambda\cdot\sum_x V_{\lambda}(x)\, \nabla_\lambda \rho^s_\lambda(x)\label{rc}\\
&=& -\int\id \lambda\cdot \left[\nabla_\lambda\sum_x V_{\lambda}(x)\, \rho^s_\lambda(x) - \sum_x\rho^s_\lambda(x) \nabla_\lambda V_{\lambda}(x) \right]
\end{eqnarray}
and
\begin{equation}
\lim_{\ve\downarrow 0} \Big\langle \int_0^{1/\ve} f^\text{exc}_{\lambda(\ve t)}(x_t)\, \id t\Big\rangle^\ve   =
\langle V_a\rangle_a -\langle V_b\rangle_b + \int \id \lambda\cdot\langle \nabla_\lambda V_\lambda\rangle_\lambda
\end{equation}
when $a, b$ are respectively the initial and final points in the protocol.  But, $\langle V_b\rangle_b =\langle V_a\rangle_a =0$ so that we end up with
\begin{equation}
\lim_{\ve\downarrow 0} \Big\langle \int_0^{1/\ve} f^\text{exc}_{\lambda(\ve t)}(x_t)\, \id t\Big\rangle^\ve   =
\int \id \lambda\cdot\langle \nabla_\lambda V_\lambda\rangle_\lambda
\end{equation}
which indeed means that the left-hand side is geometrical and in the case $f_\lambda = {\cal P}_\lambda$, that identifies the thermal-response coefficient
\begin{equation}
    D_\lambda := \big\langle\nabla_\lambda V_\lambda\big\rangle_\lambda
\end{equation}
as wished.
    \end{proof}

The equality \eqref{rc} shows that the excess heat under a quasistatic protocol equals
\begin{equation}
\lim_{\ve\downarrow 0} \Big\langle \int_0^{1/\ve} f^\text{exc}_{\lambda(\ve t)}(x_t)\, \id t\Big\rangle^\ve  = -\int\id \lambda\cdot\sum_x  V_\lambda(x)\, \nabla_\lambda \rho^s_\lambda(x).
\end{equation}
We have thus reached
\begin{equation}\label{mthe}
\lim_{\ve\downarrow 0} Q_\ve =  \int\id \lambda\cdot  \sum_x V_{\lambda}(x)
    \, \nabla_\lambda \rho^s_\lambda(x).
    \end{equation}
As a consequence, to prove Theorem \ref{t1}, the uniform boundedness of the quasipotential $V$ combined with the vanishing of $\nabla_\lambda \rho^s_\lambda(x)$ as $\beta\uparrow \infty$ suffice. We start with the boundedness of $V$ in the next section.

\subsection{Boundedness of the quasipotential}\label{proofsbp}

In this section we give a direct proof of Proposition~\ref{dec} which provides a graphical representation of the difference-quasipotential $V(\overline{e})$. An alternative proof from the matrix-forest theorem is left to Appendix~\ref{mf}.  

First we prove that the difference-quasipotential $V(\bar{e})$ given by~\eqref{treevedge}--\eqref{loopvedge} indeed defines a potential.
\begin{lemma}
    For all  loops $\ell$ in the graph $G$, and all oriented loops $\overline{\ell}$,  $\sum_{\bar{e}\in\overline{\ell}} V(\bar{e}) =0$.
    \end{lemma}
\begin{proof}
    Let $\overline{\ell}$ be the oriented loop $((x_1,x_2),x_2,(x_2,x_3), \dots , (x_{n-1},x_n),x_n,(x_n,x_1))$. We start by showing that $\sum_{\bar{e}\in\overline{\ell}} V_{tree}(\bar{e}) =0$. Put $x_{n +1} = x_1$. Then
    \begin{align}
        \sum_{\bar{e}\in\overline{\ell}} V_{tree}(\bar{e}) &= \frac{1}{W} \sum_{T\in \cal T} \sum_{u \in \cal V} w(T_u) \left( \sum_{i = 1}^n q_{T}(x_i \to x_{i +1})\right).
    \end{align}
    Since for every fixed tree $T$ and for any three vertices $x,y,z$, $q_T(x
    \to y) + q_T(y \to z) = q_T(x \to z)$, the sum $\sum_{i = 1}^n q_{T}(x_i \to x_{i +1}) = 0$ and thus $\sum_{\bar{e}\in\overline{\ell}} V_{tree}(\bar{e}) = 0$.\\
   Consider a graph $G$ which has only one loop $\ell$.  The edges  located on the loop in the underlining graph are not participating in some tree  of spanning tree-loops of the  graph $G$. Hence, $\cal H_{\overline {e}}=\emptyset$ for all $\overline{e}\in \overline{\ell}$ and then  $\sum_{\overline{e} \in \overline{\ell}}V_{loop}(\overline{e})=0$ (see \eqref{loopvedge}).\\
Assume next that the graph has more than one loop.  In that case, there are edges on a loop for which there is a chance to participate on the tree part of a spanning tree-loop. We fix an oriented loop $\overline{\ell}$ and sum over its edges. Notice that there is more than one path between every two states on the loop $\ell$. The oriented loop $\overline{\ell}$ has at least three edges.  We call two of them  $\overline{e}_1=(z,x), \overline{e}_2=(x,y)$. For the other loop there is a spanning tree-loop $H_{e_1}\not \ni e_2$ such that $e_1$ is located on a tree and $\overline{e}_1$ is leaving the loop so that $\sigma_{H_{e_1}}=+1$. Consider a tree-loop made by removing the edge $e_1$ from $H_{e_1}$ and adding the edge $e_2$. The new spanning tree-loop is denoted by $H_{e_2}\not \ni e_1$. When $\overline{e}_2$ in $H_{e_2}$ is oriented towards the loop, then $\sigma_{H}(\overline{e}_2)=-1$, while $H^{(e_1)}=H^{(e_2)}$. That scenario repeats for every edge in $\ell$: we claim that for every $H^{(e)}$ with the sign $\sigma_{H}(\overline{e})$ there is a corresponding $H^{(e')}$ with  sign  $-\sigma_{H}(\overline{e'})$ such that the tree-loop-trees are the same.  Therefore, $\sum_{\overline{e}\in \overline{\ell}}V_{\text{loop}(\overline{e})}=0$.
\end{proof}

Next, we give a graphical representation of \eqref{power} in the following
\begin{lemma}\label{aveg}
The average of ${\cal P}(x)$ in \eqref{power} is
	\begin{equation}
	\left\langle \cal P\right\rangle =\frac{1}{W}\sum_{\overline{H}\in O(\cal H) }w(\overline{H})q(\overline{\ell})
	\end{equation}
where $O(\cal H)$ is the set of all oriented spanning tree-loops in graph $G$. The  loop $\overline{\ell}$ is  the same  as the loop in $\overline{H}$. 
\end{lemma}
\begin{proof}
See Appendix \ref{avegproof}.
\end{proof}
The main step for proving Proposition \ref{dec} is the next proposition.
\begin{proposition}\label{verification}
For every $x\in \cal V$,
\begin{equation}\label{genstep}
\sum_{\overline{e} =(x,\cdot)} k(\overline{e})\,\big( -V(\overline{e}) + q(\overline{e})\big) = \langle \cal P\rangle. 
\end{equation}
\end{proposition}
\begin{proof}
    We refer to the decomposition \eqref{treevedge} and \eqref{loopvedge} of $V(\overline{e})$. We start by calculating $\sum_{\overline{e} =(x,\cdot)} k(\overline{e}) V_{\text {tree}}(\overline{e})$. If the edge $e = \{x,z\}$ is located on a spanning tree $T$, then $q_{T}(x\rightarrow z )=q(x,z)$, otherwise $\,q_{T}(x\rightarrow z )=(q_{T}(x\rightarrow z ) + q(z,x)) + q(x,z)$, where the first two terms give $q$ over an oriented loop made by the path from $x$ to $z$ on $T$ together with the edge $(z,x)$. It follows that
    \begin{align}
        \sum_{\overline{e} =(x,\cdot)} k(\overline{e}) V_{\text {tree}}(\overline{e})&=\dfrac{1}{W}\; \sum_{\overline{e} =(x,\cdot)} k(\overline{e})\sum_{y\in \cal V}\sum_{T\in \cal T} w( T_y)\,q(\overline{e} ) \notag\\&\quad +\dfrac{1}{W}\,\sum_{\overline{e} =(x,\cdot)} k(\overline{e}) \sum_{y\in \cal V}\, \, \sum_{ T\in \cal T, e \notin T} w( T_y)\,(q_{T}(x\rightarrow z )+q(z,x)).
    \end{align}
    In the first line we can use that $\sum_{y\in \cal V}\sum_{T\in \cal T} w( T_y)=W$ and the first sum is thus equal to $\sum_{\overline{e} =(x,\cdot)} k(\overline{e})q(\overline{e})$. Now look at the second line. Take $\overline{e} = (x,z)$ and  $T_y$. We look at two cases: $y = x$ and $y \neq x$. If $y = x$, then by adding the edge $\overline{e}$ to  $T_x$ an oriented  spanning tree-loop is made (but $q$ over the loop is in the opposite direction of the oriented spanning tree-loop). Every oriented spanning tree-loop with $x$ on the loop can be made in this way.\\
     If $y \neq x$ and the edge $(x,z)$ is not in $T_y$, there exists another edge such  $(x,z')$  in $T_y$. In the case that $y$ is on the path from $x$ to $z$ in the tree, by adding the edge $\overline{e} = (x,z)$ to $T_y$, again a loop is made. If $(x,z')$ is on this loop, then there exists another spanning  tree $T'$ made by removing the edge $\{x,z'\}$ from $T$ and adding the edge $\{x,z\}$ such that 
    \begin{equation}k(x,z) w( T_y)\,(q_{T}(x\rightarrow z )+q(z,x))+
    k(x,z') w( T'_y)\,(q_{T'}(x\rightarrow z' )+q(z',x))=0\end{equation}
	where $k(x,z) w( T_y)= k(x,z') w( T'_y)$. We thus have 
	\begin{align}\label{prob1}
        \sum_{\overline{e} =(x,\cdot)} k(\overline{e}) V_{\text {tree}}(\overline{e})&=\; \sum_{\overline{e} =(x,\cdot)} k(\overline{e})q(\overline{e} )-\frac{1}{W} \sum_{H,x\in \ell}\,\,\sum_{\overline{H}\in O(H)} w(\overline{H})q(\overline{\ell})\notag\\
        & \quad -\dfrac{1}{W}\,\sum_{\overline{e} =(x,\cdot)} k(\overline{e}) \sum_{\substack {y\notin \ell_{x\rightarrow z,x}}}\sum_{T\in \cal T, e \notin T} w( T_y)\,(q_{T}(z\rightarrow x )+q(x,z)) .
    \end{align}
    
    Now we look at the loop term $V_\text{loop}(\overline{e})$. Let $\overline{e} = (x,z)$. Take an arbitrary $H \in \cal H_e$. If $x$ is not on the loop part  and is not the root of the separated tree part in the oriented  tree-loop-tree $\overline{H}$, then in the graph $G$, $x$ has at least two neighbours, $z_1$ and $z_2$. Let $e_1 = \{x,z_1\}$ and $e_2 = \{x,z_2\}$, then
    \begin{equation} 
    k(x,z_1) \sum_{\substack{\overline{H}\in O(H^{(e_1)})\\ x \text{ is no root}}} \sigma_H (\overline{e_1}) q(\overline{\ell}) w(\overline{H}) + k(x,z_2) \sum_{\substack{\overline{H}\in O(H^{(e_2)})\\ x \text{ is no root}}}\sigma_H (\overline{e_2}) q(\overline{\ell}) w(\overline{H})=0
    \end{equation}
    because $\sigma_H$ has different signs. For more than two neighbours on the tree the result is the same.
    In case that $x$ is the root of the separate tree of some $\overline{H} \in O(H^{(e)})$, then by adding the edge $\overline{e}$ to $\overline{H}$, an oriented spanning tree-loop is made. We thus have
    \begin{align}\label{prob2}
        \sum_{\overline{e} =(x,\cdot)} k(\overline{e}) V_\text{loop}(\overline{e}) & = -\frac{1}{W}\sum_{\substack{H \in \cal H_{e}\\ x \notin \ell}}\, \sum_{\overline{H}\in O(H)}w(\overline{H})q(\overline{\ell})\notag\\
        &+\frac{1}{W}\sum_{\substack{\overline{e} =(x,\cdot)\\ x\in \ell}} k(\overline{e})\,\sum_{H\in \cal H_{e}}\, \sum_{\overline{H}\in O(H^{(e)})}\sigma_H(\overline{e})\, q(\overline{\ell})\,w(\overline{H}).
    \end{align}
    
    In the second line, $k(\overline{e})w(\overline{H})$ (see Fig.~\ref{fig:secondline-a}) equals $k(\overline{e'})w(T_r)$ (see Fig.~\ref{fig:secondline-b}), where $\overline{e'} = (x,z')$ with $z'$ such that $e'$ is on the oriented loop part of $\overline{H}$ and $r$ is the root of the tree part of $\overline{H}$. 
    The second line of \eqref{prob2} thus corresponds with the second line of \eqref{prob1} (see Fig.~\ref{graphicalobject}) and as a result we get 
    \begin{align}
        \sum_{\overline{e} =(x,\cdot)} k(\overline{e})\,V(\overline{e})&=\sum_z  k(\overline{e})\,(V_{\text{loop}}(\overline{e})+V_{\text{tree}}(\overline{e}))\notag\\
        &=\; \sum_{\overline{e} =(x,\cdot)} k(\overline{e})q(\overline{e} )-\frac{1}{W} \sum_{H}\,\,\sum_{\overline{H}\in O(H)} w(\overline{H})q(\overline{\ell}).
    \end{align}
    The proof is finished by using Lemma \ref{aveg}.
    
\begin{figure}[H]
     \centering
     \begin{subfigure}{0.49\textwidth}
         \centering
         \includegraphics[scale = 0.35]{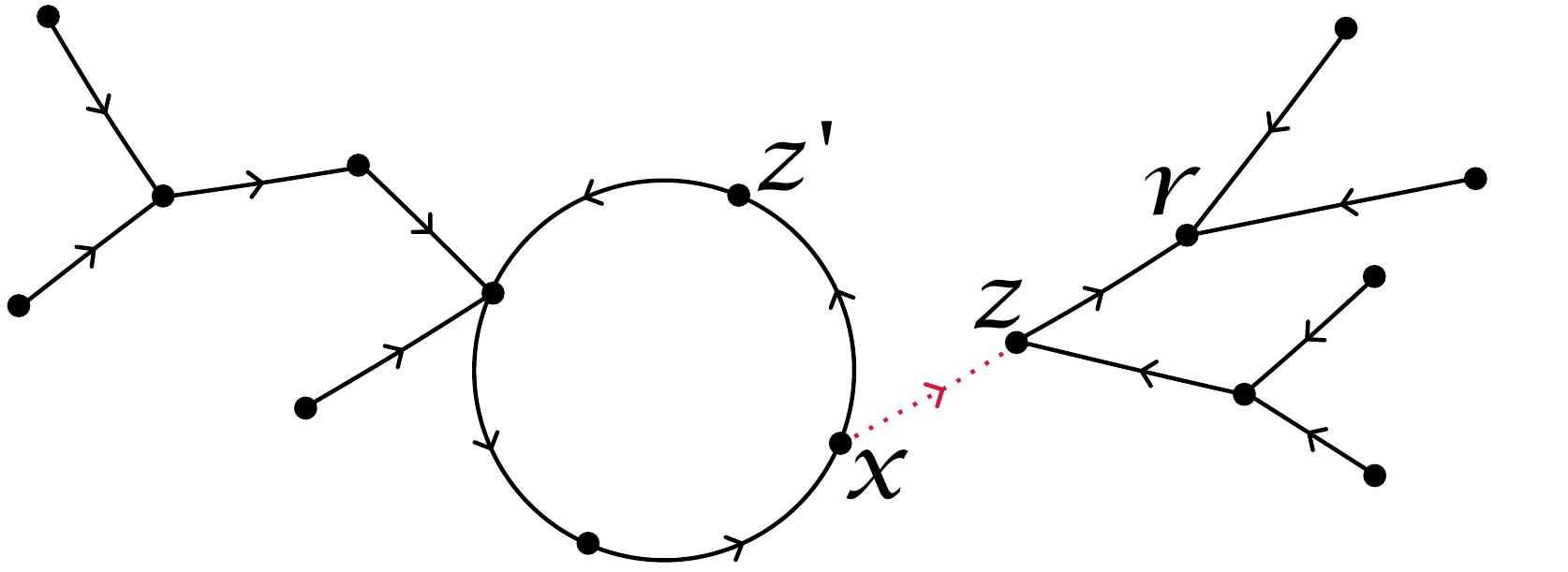}
         \caption{$\overline{e}\cup \overline{H}$}
         \label{fig:secondline-a}
     \end{subfigure}
     \hfill
     \begin{subfigure}{0.49\textwidth}
         \centering
         \includegraphics[scale = 0.35]{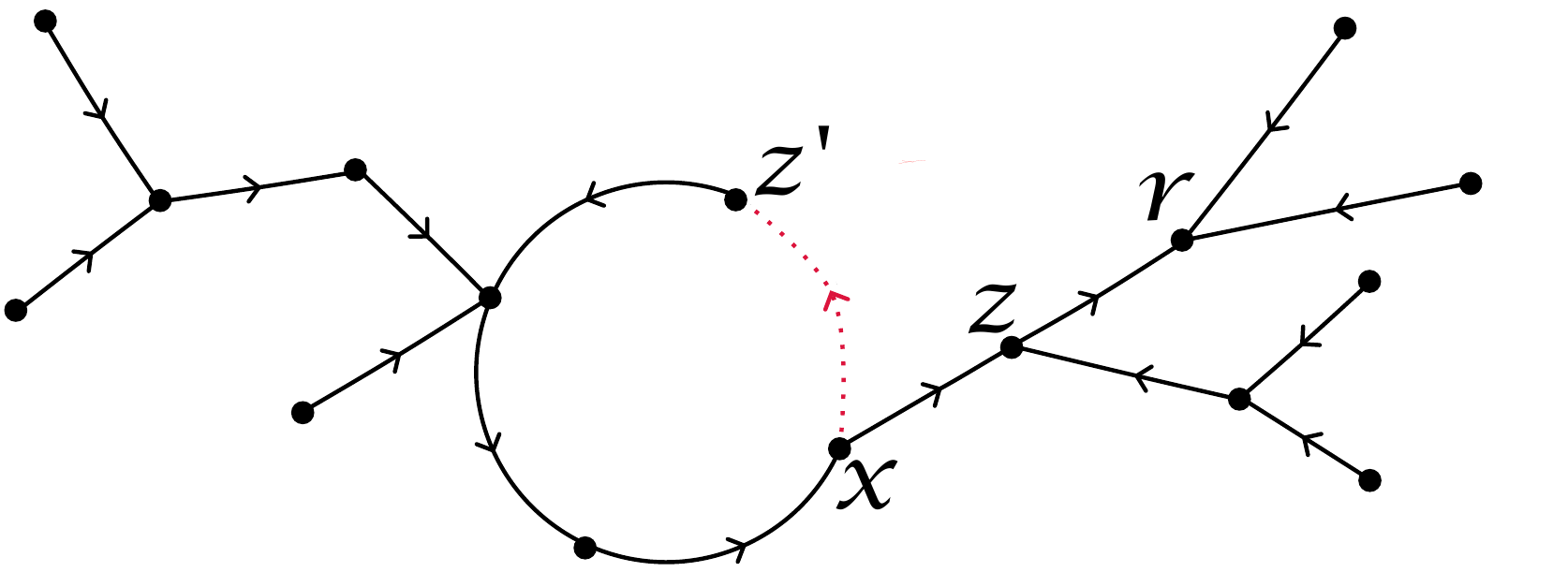}
         \caption{$\overline{e'}\cup T_r$}
         \label{fig:secondline-b}
     \end{subfigure}
        \caption{The second line of \eqref{prob2} corresponds with the second line of \eqref{prob1}.}
        \label{graphicalobject}
\end{figure}
    
\end{proof}
    
\begin{proposition}\label{boundedinedge}
    Under Condition 2, the quasipotential $V(\overline{e})$ is bounded for every edge $\overline{e}$, uniformly for $\beta \uparrow \infty$. 
\end{proposition}
\begin{proof}
Look at the tree term \eqref{treevedge} in  Lemma~\eqref{dec}.  The terms $q_T(x\rightarrow y)$ are bounded and $W \geq w(T_y)$ for all $y$.  Therefore,  $V_{\text{tree}}(\overline{e})$ is uniformly bounded in temperature.\\

 Consider next the loop term \eqref{loopvedge} in Lemma~\eqref{dec}.  There,   $\overline{H}\in O(H^{(e)})$ is  an oriented  tree-loop-tree  created by removing the edge $e$ from the spanning tree-loop $H$ and then giving  direction to the edges.  In the low-temperature asymptotics, the weights are
 \begin{equation}
  w(\overline{H})\simeq e^{\beta \phi(\overline{H})}
 \end{equation}
 and asymptotically in logarithmic sense,
 \begin{equation}
 W \simeq \sum_{T,y} e^{\beta \phi(T_y)}\simeq e^{\beta\phi^*}.
 \end{equation}
On the other hand, $q(\overline{\ell})$  is bounded.  Hence, 
 \begin{equation}
 V_{\text{loop}}(\overline{e}) \simeq e^{\beta(\phi(\overline{H})-\phi^*)}
 \end{equation}
is bounded if $\phi(\overline{H})\leq\phi^*$.
\end{proof}

Now we can finish the proof of the uniform boundedness of the quasipotential.   

\begin{proposition}\label{bdd}
Given Condition 2, the quasipotential $V_\lambda=V$ of \eqref{prv} is uniformly bounded in $\beta\uparrow \infty$.
\end{proposition}
\begin{proof}
From Proposition~\ref{boundedinedge} and by the finiteness of the graph $G$, there exists $B < \infty$ such that
$\sum_{\bar{e} \in \cal{E}} |V_\la(\bar{e})| \leq B$, uniformly in $\la$ for $\be$ large enough. By the centrality condition 
$\sum_x \rho_\la(x) V_\la(x) = 0$, for any $\la$ there exist configurations $x_0$ and $x_1$ such that
$V_\la(x_0) \leq 0 \leq V_\la(x_1)$. Therefore, for any 
$x \in \caV$,
\begin{equation}
V_\la(x) \leq V_\la(x) - V_\la(x_0) \leq \sum_{\bar{e}:\,x_0 \rightarrow x} |V_\la(\bar{e})|
\leq B 
\end{equation}
where the sum is over an arbitrary path in $G$ connecting $x_0$ and $x_1$. Analogously,
\begin{equation}
V_\la(x) \geq V_\la(x) - V_\la(x_1) \geq -\sum_{\bar{e}:\,x_1 \rightarrow x} |V_\la(\bar{e})|
\geq -B 
\end{equation}
That proves the statement.

\end{proof}

\begin{proof}[Proof of Theorems \ref{t1}--\ref{t2}]
The statements of Theorems \ref{t1}--\ref{t2} follow immediately from Proposition~\ref{condition1.2} (the vanishing of $\nabla_\alpha \rho^s_\lambda(x)$ as $\beta \uparrow \infty$), Proposition \ref{condition2} (the vanishing of $\id \rho^s_\lambda(x)/\id (\beta^{-1})$ as $\beta \uparrow \infty$) and  Proposition~\ref{bdd} (the uniform boundedness of $V_\la(x)$ along $\be\uparrow\infty$).
\end{proof}

\section{Summary and conclusion}\label{conc}

We have considered nonequilibrium Markov jump processes on a finite connected graph under the interpretation of local detailed balance.  The main object is the \emph{excess heat} $Q_\ve$ which is computed in the quasistatic limit $\ve\downarrow 0$.  We prove it gives rise to an expression in terms of the quasipotential $V_{\lambda}$ defined in \eqref{prv}: see \eqref{cl},
\begin{equation}\label{ce1}
\lim_{\ve\downarrow 0} Q_\ve =  \int\id \lambda\cdot  \sum_x V_{\lambda}(x)
    \, \nabla_\lambda \rho^s_\lambda(x).
\end{equation}
where $\rho^s_\lambda$ is the unique stationary distribution.\\
Our main result is the formulation of sufficient conditions for that quasistatic excess heat \eqref{ce1} to vanish as the temperature goes to zero. First, the low-temperature asymptotics of the stationary distribution ensures that $\nabla_\lambda \rho^s_\lambda(x) \rightarrow 0$.  Secondly, we need the boundedness of the quasipotential $V_\lambda$, uniformly in $\beta\uparrow \infty$. For that boundedness, it suffices to consider the differences of $V_\lambda$ over an edge. The result can be inferred by explicit inspection, where the condition arises more specifically from the analysis of tree-loop and tree-loop-tree contributions at low temperatures.  The main idea is however contained in the matrix-forest theorem, which we discuss in Appendix \ref{mf}.  It gives a graphical representation of the quasipotential, and it allows the evaluation of the zero-temperature limit.   \\
As a conclusion, a proof is obtained for the validity of the Nernst postulate to be extended to nonequilibrium Markov jump processes.  \\

\noindent {\bf Acknowledgment:} 
KN thanks A. Lazarescu and W. O'Kelly de Galway for previous inspiring discussions on the subject. The work was concluded while authors FK and IM visited KN at the Institute of physics  in Prague.  They are grateful for the hospitality there. 
\newpage

\appendix

\section{Matrix-forest theorem}\label{mf}
As it may seem less clear how \eqref{treevedge} and \eqref{loopvedge} arise, we give here their origin from the matrix-forest theorem.\\

Consider a continuous-time irreducible Markov process on a finite state space $\cal V$ characterized by transition rates $k(x,y)\geq 0$ for $x \, , y \in \cal V$. Let $L$ be the backward generator,
\begin{align}\label{bl}
L_{xy}&=k(x,y), \quad x\neq y \,\, \text{and } \, x,y\in \cal V\notag\\
L_{xx}&=-\sum_y k(x,y).
\end{align} 
Consider $e^{tL}g (x) = \langle g(X_t)\,|\,X_0=x\rangle$.
\begin{equation}\label{exworkint}
V_g(x) = \int_0^\infty [\langle g(X_t)\,|\,X_0=x\rangle - \langle g\rangle] \id t
\end{equation}
where we subtract the asymptotic stationary value and
\begin{equation*}
\langle g(x_t)\,|\,x_0=x\rangle\xrightarrow[]{t\uparrow \infty}\langle g\rangle.
\end{equation*}
Putting $f(x)=g(x)-\left\langle g\right\rangle$, 
\begin{equation}\label{intv}
V_g(x) = V_f(x) = \int_0^\infty e^{Lt}f(x) \,\id t
\end{equation}
so then 
\begin{equation}
V_f(x) = \lim_{b \to 0}\dfrac{-1}{(L - b I)}f(x)\end{equation}
where $\dfrac{1}{L-b I}$ is the resolvent-inverse of the backward generator $L$ (see \cite{drazin}).\\

For our purposes, the set $\cal V$ is the vertex set of a connected graph $G$, and the transitions happen over its edges.  We use \cite{che1,MForest1,MForest2} to obtain a graphical representation of $V_f$; put $V:= V_f$. A spanning forest is a collection of trees that forms a spanning subgraph.                                                                        
Define the set $\cal F^{x \to y}_m$ to be the set of all spanning forests in $G$ with $m$ edges having the properties: every tree in the forest is a rooted tree, $y$ is the root of one of the trees and $x$ and $y$ are in the same tree (so there is a path from $x$ to $y$). $F_m^{xy}$ denotes an element from the set $\cal F^{x \to y}_m$.
Define $\cal F_m$ as the union of sets  $\cal F^{x \to x}_m$ in graph $G$.

\begin{proposition}
\begin{equation}V(x) = \sum_{y} \dfrac{ w\left(\cal F_{n-2}^{x\rightarrow y}\right)}{w\left(\cal F_{n-1}\right)}f(y).\end{equation}
\end{proposition}
\begin{proof}
From \cite{che1}, for all $c>0$,
\begin{align}
\bigg( \dfrac{1}{I - c L}\bigg)_{xy} = \frac{\sum_{m=0}^{n-1} c ^m w(\cal F^{x \to y}_m)}{\sum_{m=0}^{n-1}c ^m w(\cal F_m)}
\end{align}
where $w(\cal F^{x \to y}_k)$ is the weight of the set $\cal F^{x \to y}_k$
\begin{equation}
w(\cal F^{x \to y}_m)= \sum_{F_m^{xy}\in \cal F^{x \to y}_m }\prod_{(z,z')\in F_k^{xy}}k(z,z').
\end{equation} 
Thus, 
\begin{align}
V(x) & =\sum_{y} \left(\lim_{c \rightarrow \infty} \dfrac{c}{ I-c L}\right)_{xy} f(y) \notag\\
& =\sum_{y} \lim_{c \to \infty} \frac{\sum_{m=0}^{n-1} c^{m+1}w(\cal F^{x \to y}_m)f(y)}{\sum_{m=0}^{n-1}c^m w(\cal F_m)}\notag\\
& =\sum_{y} \lim_{c \to \infty} \frac{\sum_{m=0}^{n-2} c^{m+1}w(\cal F^{x \to y}_m)f(y)}{\sum_{m=0}^{n-1}c^m w(\cal F_m)} + \sum_{y} \lim_{c \to \infty} \frac{ c^{n}w(\cal F^{x \to y}_{n-1})f(y)}{\sum_{m=0}^{n-1}c^mw(\cal F_m)}.
\end{align}
Next, use that $\left<f\right> = 0$ and hence, $\sum_{y} \rho(y)f(y) = 0$.  For the stationary distribution $\rho$ we  use the Kirchhoff formula \eqref{kir}. Therefore, $\sum_y w(\cal F^{x \to y}_{n-1})f(y) = 0$, and the second term in the last line above is equal to zero. To continue the calculation, 
\begin{align}
V(x) & = \sum_{y} \lim_{c \to \infty} \frac{\sum_{m=0}^{n-2} c^{m+1}w(\cal F^{x \to y}_m)f(y)}{\sum_{m=0}^{n-1}c^m w(\cal F_m)} \notag\\
& = \sum_{y} \dfrac{ w\left(\cal F_{n-2}^{x\rightarrow y}\right)}{ w\left(\cal F_{n-1}\right)}f(y).
\end{align}
\end{proof}

Define $\cal F^{x\rightarrow y} := \cal F_{n-2}^{x\rightarrow y}$ as the set of all spanning forests consisting of two trees. Remark that $\cal F_{n-1} $ is the set of all rooted spanning trees and $W$ is the sum over the weights of all rooted spanning trees, so then, $w(\cal F_{n-1})=W$.
\begin{corollary}
    The solution $V$ with $\langle V\rangle=0$ of $LV = -f$ with $\left< f \right>=0$ is given by 
     \begin{equation}\label{potentialforest}
        V(x) = \frac{1}{W} \sum_{y} w\left(\cal F^{x\rightarrow y}\right)f(y). 
     \end{equation}
\end{corollary}

\subsection{Quasipotential for a specific source }
Here we consider the quasipotential of \eqref{potentialforest} for a specific  source. We  focus on the  case  $f(y)=\cal P(y)-\left\langle \cal P\right\rangle $ in  \eqref{potentialforest}, where $\cal P(y)=\sum_x k(y,x)q(y,x)$ and $q(x,y)=-q(y,x)$. We define
\begin{equation}V^\cal P(x):=\dfrac{\sum_y w(\cal F^{x\rightarrow y})\cal P(y)}{W}, \quad V^{\left\langle \cal P \right\rangle}(x):= \left\langle \cal P \right\rangle \dfrac{ \sum_y w(\cal F^{x\rightarrow y})}{ W}.\end{equation}
So then, the quasipotential can be written as 
\begin{equation}\label{vf}
V(x)=V^\cal P(x)-V^{\left\langle \cal P\right\rangle}(x)
\end{equation}
 Proposition \ref{splitv} shows that the quasipotential in \eqref{vf} can be decomposed  into two terms, one related to spanning trees only and the other containing loops.\\
$F^{xy}$ denotes a forest in the set of $\cal F^{x \rightarrow y}$. Write $ k_{yx}:=  k(y,x)$ and $q(y,x):=q_{yx}$.\\
\begin{lemma}\label{forestcase}
	For each $x$ on the graph $G$,
	\begin{equation}
	\sum_y w(\cal F^{x\rightarrow y} )\cal P(y) = \sum_{z,y}	\sum_{\substack{F^{xy} \in \cal F^{x\rightarrow y} \\  (z,y)\notin F^{xy}}}\,w(F^{xy})\, k_{yz}\,q_{yz}.
	\end{equation}
\end{lemma}
\begin{proof}
     The product $w(F^{xy})\, k_{yz}$ is the weight when adding $(y,z)$ to the forest $F^{xy}$. Let us consider  forests $F^{xy}$ and $F^{xz}$ which have different directions for the edge $\{z,y\}$: $(z,y)$ is in the forest $F^{xy}\in \cal F^{x\rightarrow y}  $ and the edge $(y,z)$ is in the forest $F^{xz}\in \cal F^{x\rightarrow z}  $ .  Adding the edge $(y,z)$ to the forest $F^{xy}$ is the same as adding the edge $(z,y) $ to the forest $F^{xz}$; see Fig.~\ref{lem4}.
\begin{figure}[H]
     \centering
     \begin{subfigure}{0.49\textwidth}
         \centering
         \includegraphics[scale = 0.5]{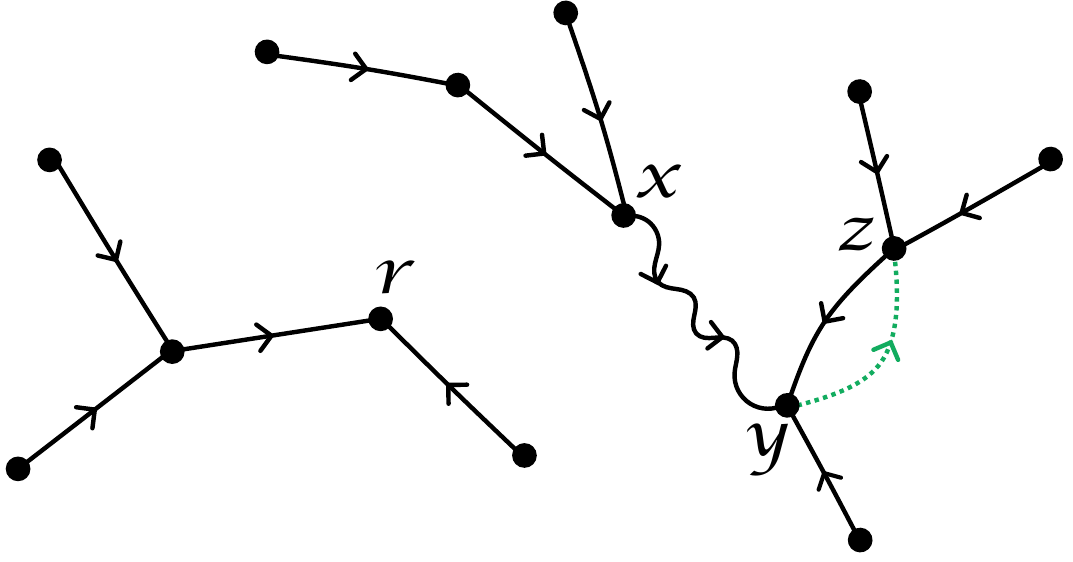}
         \caption{$F^{xy} \cup (y,z)$}
     \end{subfigure}
     \hfill
     \begin{subfigure}{0.49\textwidth}
         \centering
         \includegraphics[scale = 0.5]{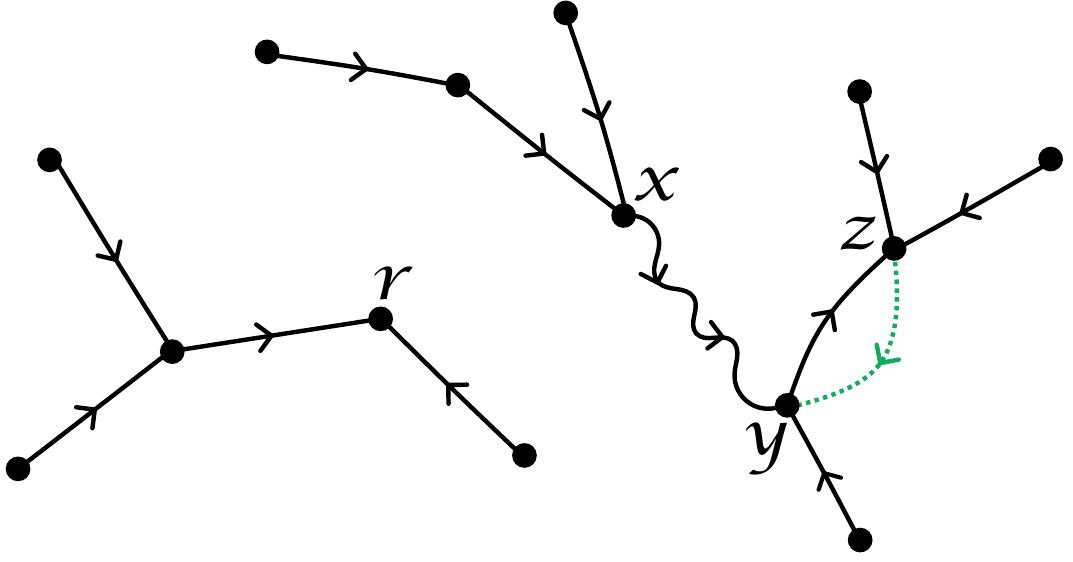}
         \caption{$F^{xz} \cup (z,y)$}
     \end{subfigure}
        \caption{\small{Adding the edge $(y,z)$ to $F^{xz}  $ and adding the edge $(y,z)$ to the forest $F^{xy}$.}}\label{lem4}
\end{figure}
and
 \begin{equation}w(F^{xy})\,k(y,z)q(y,z)\, +\,  w(F^{xz})\,k(z,y)q(z,y)=0.
 \end{equation}
So then
\begin{align}
 \sum_y w(\cal F^{x\rightarrow y} )\cal P(y) &=\sum_y w(\cal F^{x\rightarrow y} )\sum_zk_{yz}q_{yz}\notag\\
 &=\sum_{z,y}	\sum_{\substack{F^{xy} \in \cal F^{x\rightarrow y} \\  (z,y)\notin F^{xy}}}\,w(F^{xy})\, k_{yz}\,q_{yz}.
\end{align}
\end{proof}

\begin{lemma}\label{twocases}
By adding the edge $(y,z) \in G$ to the forest $F^{xy} \in \cal F^{x\rightarrow y}$ where $(z,y)\notin F^{xy}$, the new graph is either a  rooted spanning tree or an  oriented  tree-loop-tree.
\end{lemma}
\begin{proof}
Consider the forest $F^{xy}$.  It has two trees: one is rooted in $y$ and the other tree is rooted in some vertex $r$. Vertices $x$ and $y$ are located in a same tree and the edge $(z,y)$ is not in the forest. If the vertex $z$ is on the same tree with $y$, adding the edge $(y,z)$  creates  an oriented  tree-loop-tree,  Fig.~\ref{lem}(a), and if $z$ is in another tree then  a rooted  spanning tree is created,  Fig.~\ref{lem}(b). 
\begin{figure}[H]
     \centering
     \begin{subfigure}{0.49\textwidth}
         \centering
         \includegraphics[scale = 0.5]{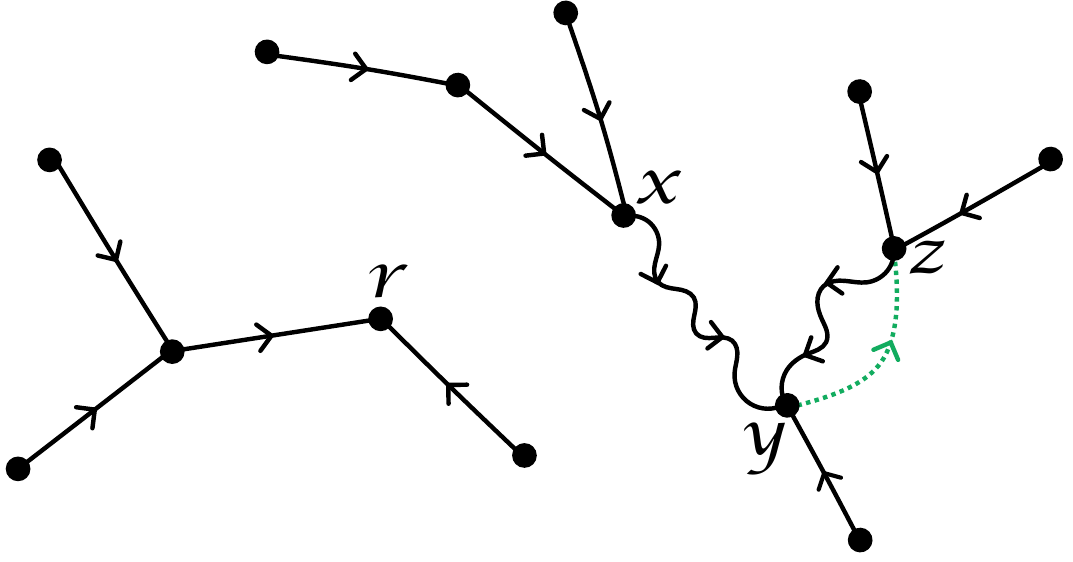}
         \caption{$F^{xy} \cup (y,z)$}
     \end{subfigure}
     \hfill
     \begin{subfigure}{0.49\textwidth}
         \centering
         \includegraphics[scale = 0.5]{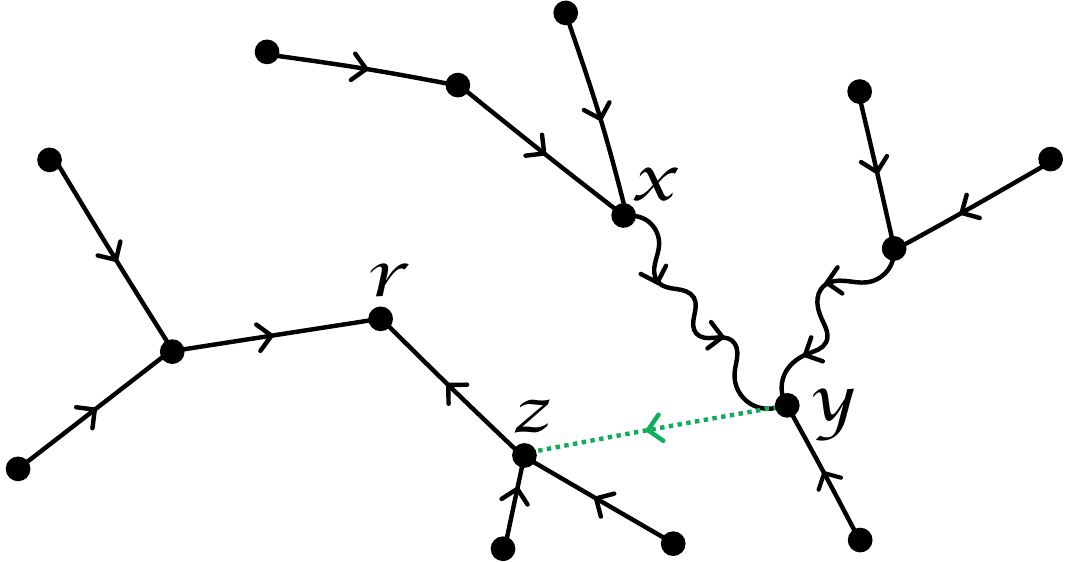}
         \caption{$F^{xy} \cup (y,z)$}
     \end{subfigure}
    \caption {\small{The edge $(y,z)$ is added to the forest $F^{xy}$, either an oriented tree-loop-tree is created (a) or a spanning rooted tree (b).} }\label{lem}
\end{figure}
\end{proof}

We need extra notation for the next Lemma.  Put $\cal K_x$ for the set of all   tree-loop-trees such that $x$ is located on the tree-loop part. If $\kappa \in \cal K _x$, then $O(\kappa)$ is the set of oriented  tree-loop-trees for all possible directions in $\kappa$ and $\overline{\kappa}$ denotes an element from the set $O(\kappa)$. $\cal L$ is the set of all non-oriented loops in the graph $G$ and $\cal K_{\ell,x}$ denotes the set of all   tree-loop-trees including the non-oriented loop $\ell$ and $x$ in the tree-loop part. 
\begin{lemma}\label{lemma2}
	 $V^\cal P(x)$ splits into two parts; one where we only sum over spanning trees and one where loops are present:
	\begin{equation}\label{split}
	V^\cal P(x)= V^\cal P_{\text{tree}}(x) + V^\cal P_{\text{loop}}(x).
	\end{equation}
	Here, the first part is equal to 
\begin{align}
V^\cal P_{\text{tree}}(x) = \dfrac{1}{W}\; \sum_{y\in \cal V}\sum_{T\in \cal T } w(T_y)\,q_{T }(x\rightarrow y )
\end{align}
 and
\begin{equation}
q_{T }(x\rightarrow y ):= \sum_{(u,u')\in (x\xrightarrow{T }y)}q(u,u')
\end{equation} 
where $(x\xrightarrow{T }y)$ is the path from $x$ to $y$ in the spanning tree $T$. The second part is equal to
\begin{align}
V^\cal P_{\text{loop}}(x) =\dfrac{1}{W}\,\sum_{\ell\in \cal L}\sum_{\kappa \in\cal K_{\ell,x}}\,\sum_{\overline{\kappa} \in O(\kappa ) }w(\overline{\kappa})q(\overline{\ell} )
\end{align}
where $\overline{\ell}$ is  same  as the loop in $\overline{\kappa}$ and $q(\overline{\ell})=\sum_{(u,u')\in \overline{\ell}}q(u,u')$.
\end{lemma}
\begin{proof}
Consider the first case in Lemma \ref{twocases}; there are different possible  tree-loop-trees for every  loop in the graph $G$ and it follows there are different possible oriented tree-loop-tree for every oriented loop. Take an  edge $(u,u')$ located on the  (oriented) loop of the oriented tree-loop-tree, rremoving the edge $(u,u')$ from the oriented tree-loop-tree gives a forest where $u$ and $u'$ are in the same connected component. There is a forest $F^{xu}$ such that $u'\in F^{xu} $ and $(u,u') \notin F^{xu}$, and 
\begin{equation}
\sum_{(u,u')\in \ell}w(F^{xu})k_{uu'}q_{uu'}=w(\overline{\kappa})q(\overline{\ell}).
\end{equation}
 Let us go back to Lemma \ref{twocases}.  For a rooted spanning  tree $T_y$ in graph $G$, there is a unique path $(x\xrightarrow{T }y)$ from $x$ to $y$ on this tree.  Then, 
\begin{equation}\label{treeforest}
    w(T_y)q_{T}(x\rightarrow y)=\sum_{(u,u')\,\in (x\xrightarrow{T }y)}w(F^{xu})k_{uu'}q_{uu'}
\end{equation}
where $F^{xu}\in \cal F^{x\rightarrow u}$. We have used the fact that by removing  an edge $\{u,u'\}$ located on the path  $(x\xrightarrow{T }y)$ in spanning tree $T_y$, a forest including two trees, $ \tau_u$ (toward $u$) and $\tau_y $, is created. Obviously, the edge $\{u,u'\} \notin F^{xu} $ and there is no other path between $u$ and $u'$ in the new forest. So then, the sum over all possible oriented  spanning trees in the left-hand side of  relation \eqref{treeforest} will give the weight of  all possible forests such that 
\begin{align}
  \sum_{y\in \cal V}\sum_{T \in \cal T}w(T_y)q_{T}(x\rightarrow y)&=\sum_{y\in \cal V}\sum_{z\in A_{F^{xy}}}w(\cal F^{x\rightarrow y})k_{yz}q_{yz} 
\end{align}
where
\begin{equation} A_{F^{xy}} = \{z \in \cal V ; \quad \nexists \,  (z\xrightarrow{F^{x y}}y) \in  F^{x y}\} \end{equation}
is the set of all vertices  from which there is no path to $y$ in the forest $F^{xy}$. Hence, according to the first case in Lemma \ref{twocases}, the tree-term of $V^\cal P$ is
\begin{equation}
V^\cal P_\text{tree}(x) = \dfrac{1}{W}\; \sum_{y\in \cal V}\sum_{T\in \cal T} w(T_y)\,q_{T}(x\rightarrow y ).
\end{equation}
\end{proof} 
We are ready for the main result of this Appendix.
\begin{proposition}\label{splitv}
	The graphical representation of the quasipotential in \eqref{potentialforest} for  
	\begin{equation} 
	f(y)=\cal P(y)-\left\langle \cal P \right\rangle,  \quad \cal P(y)=\sum_x k(y,x)q(y,x), \quad q(x,y)=-q(y,x)
	\end{equation} 
	consists of two distinct classes of terms, trees and loops: 
	\begin{align*} 
	V(x) &= V_{\text {tree}}(x)+V_{\text{loop}}(x)
	\end{align*}
	where 
	\begin{align} 
	& V_\text{tree}(x) = \dfrac{1}{W}\; \sum_{y\in \cal V}\sum_{T\in \cal T} w(T_y)\,q_{T}(x\rightarrow y ) \label{vtree}\\
	&V_{\text{loop}}(x)=\dfrac{1}{W}\sum_{\ell \in \cal L}\,\sum_{\kappa \in\cal K_{\ell,x}}\,\sum_{\overline{\kappa} \in O(\kappa ) }w(\overline{\kappa})q(\overline{\ell} )
 -\frac{1}{W}\sum_y w(\cal F^{x\rightarrow y})\left\langle \cal P \right\rangle\label{vloop}.
	\end{align}
\end{proposition}

\begin{proof}
From Eq.~\eqref{vf} we get that 
\begin{equation}
V(x) = V^\cal P(x) - V^{\left<\cal P\right>}(x).
\end{equation}
Now use the representation of $V^\cal P(x)$ from Lemma \ref{lemma2}. The representation of $\left\langle \cal P \right\rangle$ is given in Lemma \ref{aveg}. 
\end{proof}

\begin{lemma}\label{diffforest}
Consider a connected graph $G$ such that the  edge $(x,x')\in \cal E(G)$ then 
\begin{equation}
\sum_z w(\cal F^{x\rightarrow z})-\sum_y w(\cal F^{x'\rightarrow y})=0.
\end{equation}
\end{lemma}
\begin{proof}
 Take arbitrary $z$ and $F \in \cal F^{x \to z}$. Then $F$ consists of two disconnected trees $\tau_1$ and $\tau_2$, where $\tau_1$ is a tree rooted in $z$, $x \in \tau_1$ and $\tau_2$ is a tree rooted in some vertex $r$. There are two possibilities, $x'\in \tau_1$ or $x' \in \tau_2$. If $x' \in  \tau_1$, then also $F \in \cal F^{x' \to z}$. If $x' \in \tau_2$, then $F \in \cal F^{x'\to r}$. In the same way, if we fix $y$, every forests $F \in \cal F^{x' \to y}$ corresponds to a forests $F \in \cal F^{x \to z}$ for some $z$. This is a one-to-one correspondence.
\end{proof}

\subsection{Difference-quasipotential on an edge}
Consider an edge $e:=\{x,y\}$.  We write $\overline{e}:=(x,y)$ when it has a direction. Put $V(\overline{e}):=V(x)-V(y)$.  From Proposition \eqref{splitv}, 
\begin{align}
V(\overline{e})=  V_{\text {tree}}(\overline{e})+ V_{\text{loop}}(\overline{e}).
	\end{align}
Recall \eqref{vtree} 
\begin{align}
    V_\text{tree}(\overline{e}) &= \dfrac{1}{W}\; \sum_{u\in \cal V}\sum_{T\in \cal T} w(T_u)\,q_{T}(x\rightarrow u )-\dfrac{1}{W}\; \sum_{u\in \cal V}\sum_{T\in \cal T} w(T_u)\,q_{T}(y\rightarrow u )\notag\\
    &=\dfrac{1}{W}\; \sum_{u\in \cal V}\sum_{T\in \cal T} w(T_u)\big(q_{T}(x\rightarrow u )+q_{T}(u\rightarrow y )\big)\notag\\
    &=\dfrac{1}{W}\; \sum_{u\in \cal V}\sum_{T\in \cal T} w(T_u)\,q_{T}(x\rightarrow y ).
\end{align}
According to Proposition \ref{splitv} and  Lemma \ref{diffforest} the difference of loop terms is given as follows:

\begin{align}\label{diffloopv1}
V_{\text {loop}}(\overline{e})=\dfrac{1}{W}\sum_{\ell \in \cal L}\sum_{\kappa \in\cal K_{\ell,x}}\,\sum_{\overline{\kappa} \in O(\kappa) } w(\overline{\kappa} )q(\overline{\ell} )-\dfrac{1}{W}\sum_{\ell \in \cal L}\sum_{\kappa \in\cal K_{\ell,y}}\,\sum_{\overline{\kappa}\in O(\kappa ) }w(\overline{\kappa} )q(\overline{\ell} )
	\end{align}
where $\cal K_{\ell,x}$ denotes the set of all  tree-loop-trees such that $x$ is located in the tree-loop part. Take one tree-loop-tree, if $x$ and $y$  both are in the tree-loop part, then the result of equation \eqref{diffloopv1} is zero. We continue with  the case that in tree-loop-trees, $x$ and $y$ are located in different parts, one of them is located in the tree-loop part and the other is located in the tree part:
\begin{align}\label{diffloop2}
V_{\text{loop}}(\overline{e})=\dfrac{1}{W}\sum_{\ell \in \cal L}&\bigg(\sum_{\kappa\in\cal K_{\ell,x|y}}\,\sum_{\overline{\kappa} \in O(\kappa) }w(\overline{\kappa} )q(\overline{\ell} )-\sum_{\kappa\in \cal K_{\ell,y|x}}\,\sum_{\overline{\kappa} \in O(\kappa) }w(\overline{\kappa} )q(\overline{\ell} )\bigg).
	\end{align}
Here, $\cal K_{\ell,x|y}$ denotes a   tree-loop-tree with  loop $\ell$ such that $x$ is located in the tree-loop part and $y$ is located in the tree part. In the case that the edge $e$ is not located on a tree-loop $V_{\text{loop}}(\overline{e})=0$.  Put $\cal H_{\ell,e}$ as the set of all spanning tree-loops (including the loop $\ell$)  such that the edge $e$ is located on a tree and $\cal H_{\ell }^{(e)}$ is the set of all  tree-loop-trees which are created by removing the edge $e$ from a spanning tree-loop of $\cal H_{\ell,e}$. Corresponding to what state of the edge   $e$ is closer to the loop, the components of $\cal H_{\ell,e}$ can be split in two groups. The set of spanning tree-loops where $x$ is closer to the loop is denoted by $\cal H_{\ell_x, e }$, while $\cal H^{(e)}_{\ell_x }$  is  the set of all  tree-loop-trees such that $x$ is located on the tree-loop part. Rewrite \eqref{diffloop2} as
	\begin{align}
V_{\text {loop}}(\overline{e})=\dfrac{1}{W}\bigg(\sum_{\ell \in \cal L}\sum_{H \in \cal H^{(e)}_{\ell_x }}\sum_{\overline{ H}\in O(H) }w(\overline{ H})q(\overline{\ell} )-\sum_{\ell \in \cal L}\sum_{H \in \cal H^{(e)}_{\ell_y }}\sum_{\overline{ H}\in O(H) }w(\overline{ H})q(\overline{\ell} )\bigg).
	\end{align}
	
 Notice that  for every spanning tree-loop including the edge $e$ on a tree either $ \cal H^{(e)}_{\ell_x } $ or $ \cal H^{(e)}_{\ell_y } $ happens. \\

As an example, consider the graph in Fig.~\ref{exgraph} which has three loops. The possible tree-loops are shown in Fig.~\ref{extl}.

To find the differences of quasipotentials over the edge $\overline{e}=(x,w)$ the tree-loops including the edge $e:\{x,w\}$ in a tree are engaging, see $H_2$ and $H_3$ in Fig.~\ref{extl}. 

We first look at the  unoriented tree-loop  $H_{\ell,e}$ and we remove the edge $e$. In that way a tree-loop-tree is created.  Secondly, we consider  different possible orientations for the created  tree-loop-trees.  We rewrite \eqref{diffloop2},
\begin{align}
V_{\text{loop}}(x,y)=\frac{1}{W}\sum_{\ell \in \cal L}\,\sum_{H\in \cal H_{\ell,e}}\,\sum_{\overline{H}\in O(H)}\sigma_H(\overline{e})\, q(\overline{\ell})\,w(\overline{H})
	\end{align}
where $\sigma_H(\overline{e})=\pm 1$. If the edge $\overline{e}$ is oriented towards the loop of $H$, then  $\sigma_H(\overline{e})=-1$; otherwise it is positive. So then,

\begin{align}\label{diffloop3}
V_{\text{loop}}(\overline{e})=\frac{1}{W}\,\sum_{H\in \cal H_{e}}\, \sum_{\overline{H}\in O(H^{(e)})}\sigma_H(\overline{e})\, q(\overline{\ell})\,w(\overline{H}),\quad e\in \text{trees}
\end{align}
where $\cal H_e$ is a set of all spanning  tree-loops including the edge $e$ in a tree. If $H\in \cal H_e$, then $H^{(e)}$ denotes a tree-loop-tree which is created by  removing the edge $e$ from the tree-loop $ H$. $O(H^{(e)})$ is the set of all oriented  tree-loop-trees made by giving all possible orientations to $H^{(e)}$ (remember Definition \ref{tlt}).
Finally, the difference of the quasipotential over a directed edge $\overline{e}$ is

\begin{align}
V(\overline{e})&=  V_{\text {tree}}(\overline{e})+ V_{\text{loop}}(\overline{e})\notag\\
&=\dfrac{1}{W}\; \sum_{u\in \cal V}\sum_{T\in \cal T} w(T_u)\,q_{T}(x\rightarrow y )\notag\\
&+\frac{1}{W}\,\sum_{H\in \cal H_{e}}\, \sum_{\overline{H}\in O(H^{(e)})}\sigma_H(\overline{e})\, q(\overline{\ell})\,w(\overline{H}),\quad e\in \text{trees}.
	\end{align}
 
\section{Proof of Lemma \ref{aveg}}\label{avegproof}
\begin{lemma}
	For any directed edge $\overline{e}=(x_e,y_e)$, the average of $\cal P(x_e)=\sum_{y_e} k(\overline{e})q(\overline{e})$ is
	\begin{equation}
	\left\langle \cal P \right\rangle =\sum_{x_e}\sum_{y_e}k(\overline{e})q(\overline{e})\rho (x_e)=\frac{1}{W}\sum_{\overline{H} \in O(\cal H)}w(\overline{H})q(\overline{\ell})
	\end{equation}
\end{lemma}
where $w(\overline{H})=\prod_{(z,z')\in \overline{H}}k(z,z')$.
\begin{proof}
Take a rooted spanning tree $T_{x_e}$ and the edge $\overline{e}=(x_e,y_e)$. We consider two cases:\\
\emph{First case,} if the edge $\overline{e}'=(y_e,x_e)\in T_{x_e}$, then 
\begin{equation}k(\overline{e})w(T_{x_e})q(\overline{e})=k(\overline{e}')w(T_{y_e})q(\overline{e})=-k(\overline{e}')w(T_{y_e})q(\overline{e}').\end{equation}
From here we can conclude that 
\begin{equation}\sum_{T\in \cal T}\sum_{\overline{e},\overline{e}' \in T_{x_e}}k(\overline{e})w(T_{x_e})q(\overline{e})=0.\end{equation}
\emph{Second case,} if the edge $\overline{e}'=(y_e,x_e)\notin T_{x_e}$, then
$k(\overline{e})w(T_{x_e})$
is the same as the weight of a graphical object made by adding the edge $\overline{e}$ to the rooted tree $T_{x_e}$. That graphical object is a tree-loop. As  a consequence, 
\begin{equation}
k(\overline{e})w(T_{x_e})q(\overline{e})=w(\overline{H})q(\overline{e})
\end{equation}
where the loop of $\overline{H}$ is in the  same direction as $\overline{e}$. Summing over the edges in loop $\overline{\ell}$ (in clockwise or counter clockwise direction) in the underlying graph $G$ gives
\begin{equation}
\sum_{ \overline{e} \in \overline{\ell}} k(\overline{e})w(T_{x_e})q(\overline{e})=w(\overline{H}_{\ell})q(\overline{\ell}).
\end{equation}
Finally, we obtain
\begin{align}
   \sum_{x_e}\sum_{y_e}k(\overline{e})q(\overline{e})\rho (x_e)&=\frac{1}{W}    \sum_{x_e}\sum_{y_e}\sum_{T\in \cal T} k(\overline{e})q(\overline{e}) w(T_{x_e})\notag\\
   &=\frac{1}{W} \sum_{\overline{e}}\sum_{T\in \cal T} k(\overline{e})q(\overline{e}) w(T_{x_e})\notag\\
   &= \frac{1}{W} \sum_{\overline{e}}\sum_{T\ni\overline{e}' } k(\overline{e})q(\overline{e}) w(T_{x_e})+ \sum_{\overline{e}}\sum_{T\not\owns  \overline{e}' } k(\overline{e})q(\overline{e}) w(T_{x_e})\notag\\
   &=\frac{1}{W} \sum_\ell w(\overline{H}_{\ell})q(\overline{\ell})\notag\\
   &=\frac{1}{W}\sum_{\overline{H} \in O(\cal H)}w(\overline{H})q(\overline{\ell}).
\end{align}

\end{proof}

\section{Examples and illustrations}\label{exam}
This section is meant to clarify the graphical conditions of the main Theorems \ref{t1}--\ref{t2}.  It is not essential for the proofs and it can be skipped at first reading.  We first illustrate the notation with some simple examples.\\

\begin{example}\label{tree}
 Consider the graph $G$ in Fig.~\ref{ex}.
\begin{figure}[H]
    \centering
    \includegraphics[scale=0.35]{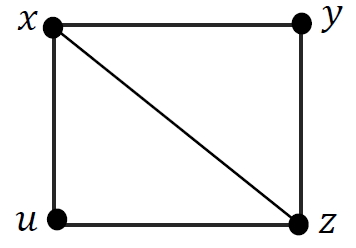}
   \caption{Graph $G$ made by three loops.}
    \label{ex}
    \end{figure}
The transition rates are
\begin{align}
 k(y,x)=k(y,z)=k(x,u)=k(u,z)=e^{-\beta}
\end{align}
and the rates over all other edges are equal to one. Definition \ref{phi} gives 
 \begin{align}
 &\phi(y,x)=\phi(y,z)=\phi(x,u)=\phi(u,z)=-1, \notag\\
 &\phi(x,y)=\phi(z,y)=\phi(u,x)=\phi(z,u)=\phi(z,x)=\phi(x,z)=0.
\end{align}
To find $\phi^*$ we need to  look at all  rooted spanning trees in the graph. The spanning trees are in Fig.~\ref{ext}.
        \begin{figure}[H]
    \centering
    \includegraphics[scale=0.4]{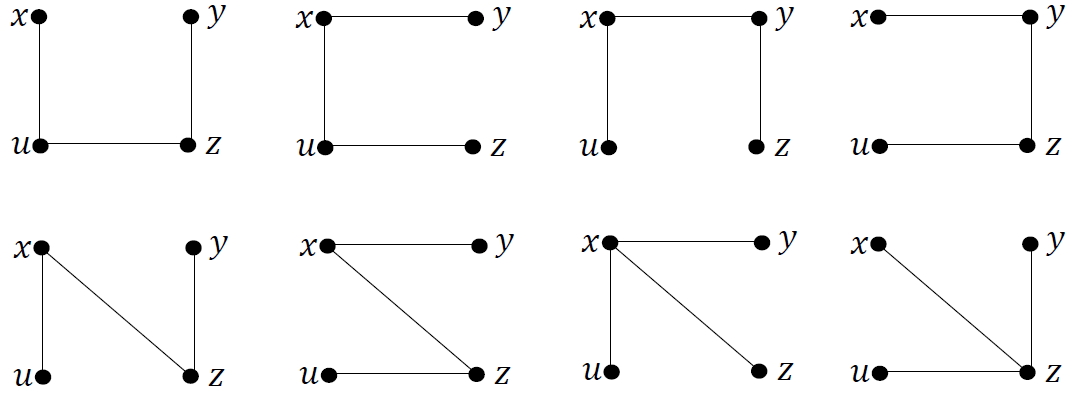}
   \caption{Spanning trees of the graph in Fig.~\ref{ex}. See \cite{intr} for more details.}
    \label{ext}
    \end{figure}
We find 
 \begin{align}
 \phi(x)=\phi(z)=\phi(u)=-1 , \quad \phi(y)=0
\end{align}
and thus $\phi^*=0$. Denote 
     \begin{equation}e_1:=\{x,y\}, \quad e_2:=\{y,z\}, \quad e_3:=\{z,u\}, \quad  e_4:=\{x,u\}, \quad  e_5:=\{x,z\}.\end{equation}
    To check  Condition 1a  we also  need to find all spanning tree-loops in the graph which are shown in Fig.~\ref{ex-tl}.
        \begin{figure}[H]
    \centering
    \includegraphics[scale=0.8]{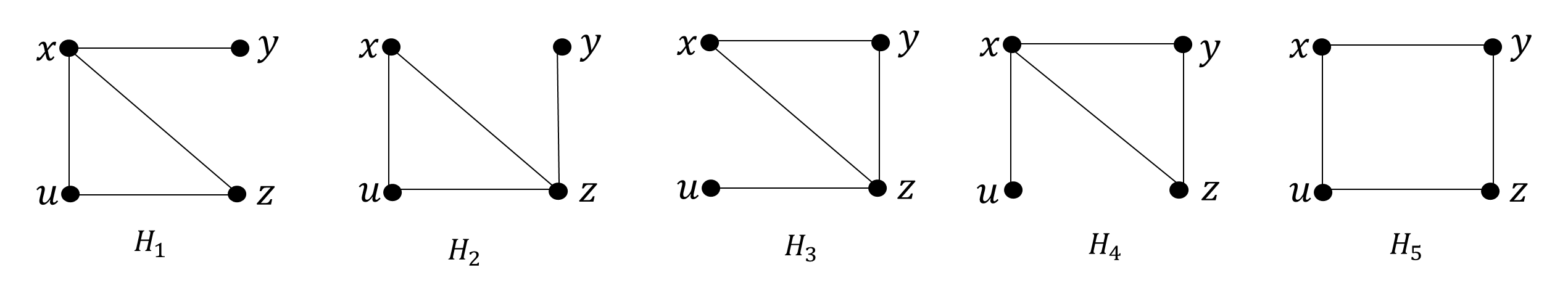}
   \caption{All tree-loops in the graph of Fig.~\ref{ex}.}
       \label{ex-tl}
    \end{figure}
For every edge, we need to look at the set of all spanning tree-loops that include this edge on a tree, 
   \begin{equation}\cal H_{e_1}=\{H_1\}, \quad \cal H_{e_2}=\{H_2\}, \quad \cal H_{e_3}=\{H_3\}, \quad \cal H_{e_4}=\{H_4\}, \quad \cal H_{e_5}=\emptyset. \end{equation}
    The next step is to construct for every possible edge the set of all oriented tree-loop-trees $O(\cal H^{(e)})$ (but that is not possible for the edge $e_5$). For example, if we look at the edge $e_1$, we need the set $O(H^{(e_1)}_1)$, which consists of two elements depending on the two possible orientations in the loop of $H_1$ see Fig.~\ref{ex-tl1}. 
       \begin{figure}[H]
    \centering
    \includegraphics[scale=0.20]{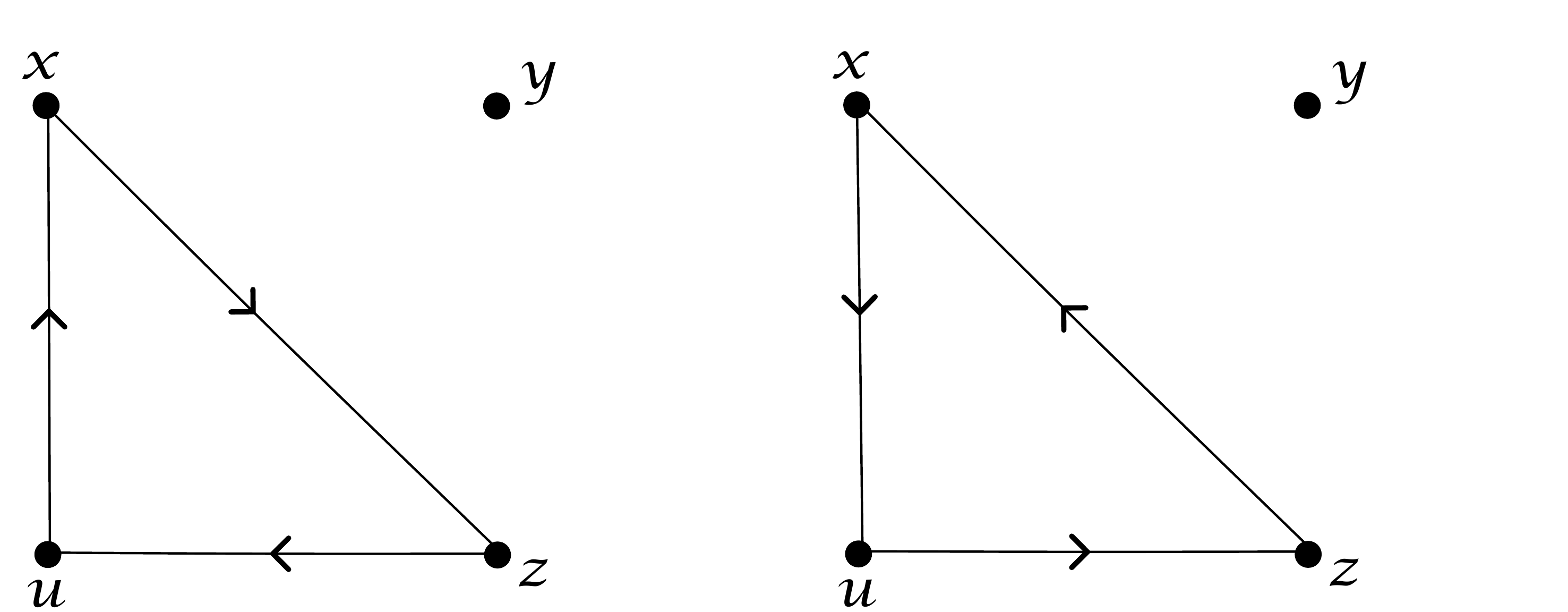}
   \caption{ The possible oriented tree-loop-trees created by removing the edge $e_1$ from  $H_{1}$.}
       \label{ex-tl1}
    \end{figure} 
Therefore, $O(H_1^{(e_1)})=\{\overline{H}_{1,1},\overline{H}_{1,2}\}$, where $\phi(\overline{H}_{1,1})=0 $ and $\phi(\overline{H}_{1,2})=-2$. We now do the same for the other edges and get, 
     \begin{align}
O(H_2^{(e_2)})&=\{\overline{H}_{2,1},\overline{H}_{2,2}\}, \quad \phi(\overline{H}_{2,1})=0,\, \phi(\overline{H}_{2,2})=-2\notag\\
O(H_3^{(e_3)})&=\{\overline{H}_{3,1},\overline{H}_{3,2}\}, \quad \phi(\overline{H}_{3,1})=-1,\, \phi(\overline{H}_{3.2})=-1\notag\\
O(H_4^{(e_4)})&=\{\overline{H}_{4,1},\overline{H}_{4,2}\}, \quad \phi(\overline{H}_{4,1})=-1,\, \phi(\overline{H}_{4,2})=-1.
\end{align}
We see that Condition 1a is satisfied.\\
To check  Condition 1b, we can use the Kirchhoff formula which expresses the stationary distribution in terms of weights of rooted spanning  trees; see e.g. \cite{intr}.  Here,  the state $y$ is the unique dominant state: $x^*=y$, unique maximizer of $\phi(z)$ in \eqref{phi}. 
\end{example}

Next we give a counterexample to  Condition 2, which can be interpreted as a ``no-delay condition''; for details of such a physical interpretation see the earlier discussion in \cite{jir} and the recent \cite{nernst}.

\begin{example}\label{nod}
Consider the graph in  Fig.~\ref{ex2} made by a centered triangle such that each state is symmetrically connected to the state in the center. 
\begin{figure}[H]
    \centering
    \includegraphics[scale=1]{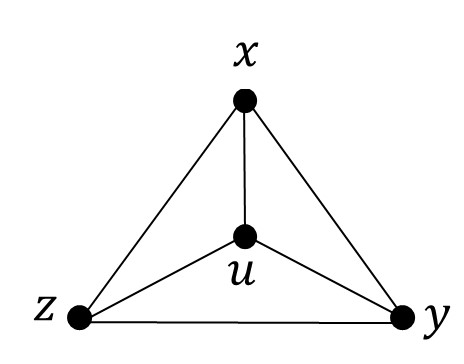}
   \caption{The graph has four loops.}
    \label{ex2}
    \end{figure}
A random walker is moving on the graph with transition rates 
\begin{align}
&k(x,y)=k(y,z)=k(z,x)= \nu, \quad k(y,x)=k(x,z)=k(z,y)= \nu e^{-\beta \ve}\notag\\
&k(x,u)=k(y,u)=k(z,u)=\mu\,e^{-\beta a}, \quad k(u,x)=k(u,y)=k(u,z)=\mu\,e^{-(a+\Delta)\beta} 
\end{align}
where $\varepsilon, \, a, \, \Delta \geq 0$.
The spanning trees with the largest weight are rooted in $u$; they start from a state on the triangle and visit the next state on the outer triangle in clockwise direction before going to the center. Therefore, $\phi^*=-a$.\\
There are $12$ spanning tree-loops, shown in Fig.~\ref{ex2tl}.

\begin{figure}[h!]
    \centering
    \includegraphics[scale=1]{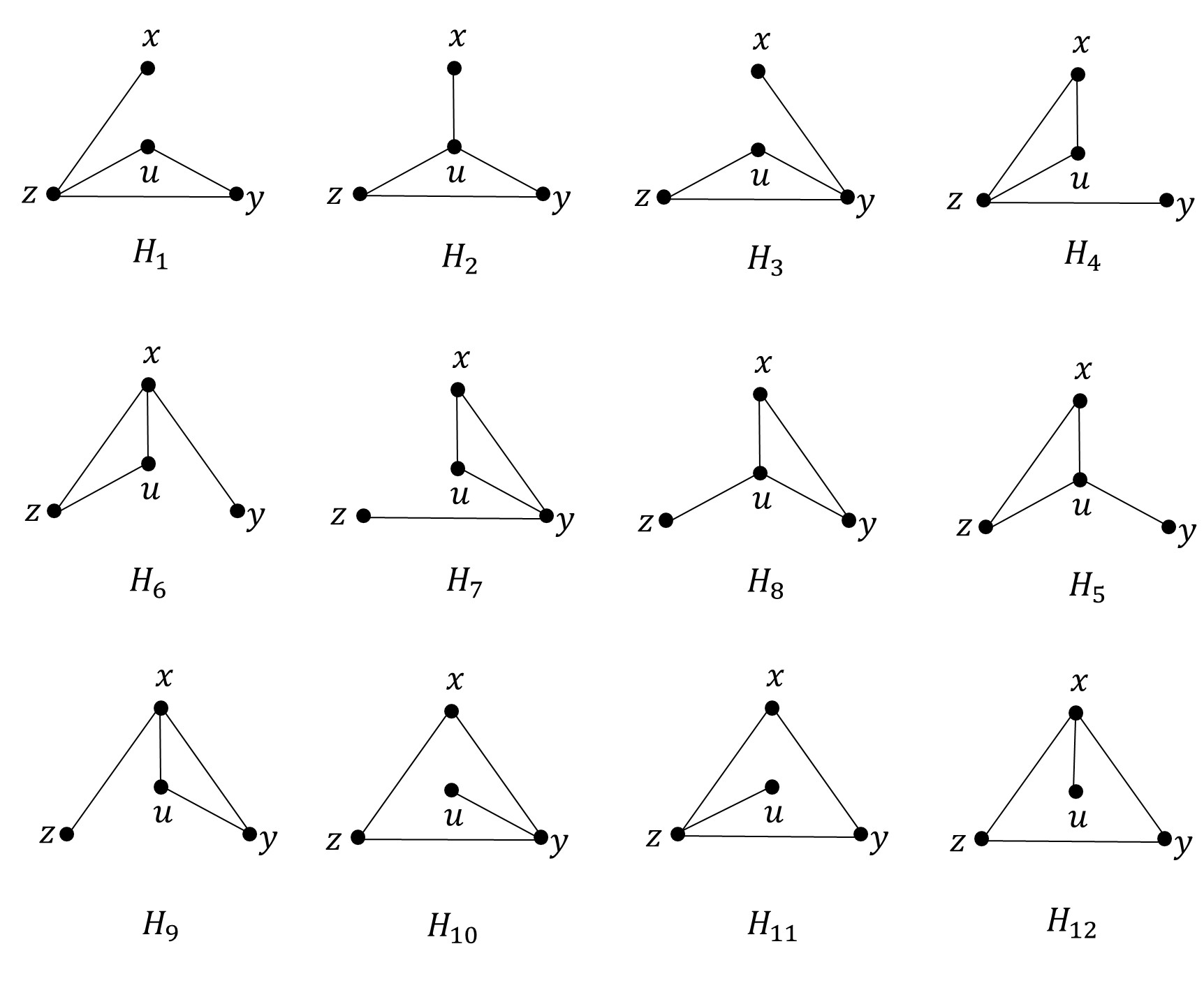}
   \caption{All spanning tree-loops in the graph of Fig.~\ref{ex2}.}
    \label{ex2tl}
    \end{figure}
    
Putting
\begin{align}
    e_1:=\{x,y\}, \quad e_2:=\{y,z\}, \quad e_3:=\{z,x\}, \quad  e_4:=\{x,u\},\quad  e_5:=\{y,u\},\quad  e_6:=\{z,u\}
\end{align}
we have
\begin{align}
    &\cal H_{e_1}=\{H_3,H_6\}, \quad \cal H_{e_2}=\{H_4,H_7\}, \quad \cal H_{e_3}=\{H_1,H_9\}, \notag\\
    &\cal H_{e_4}=\{H_2,H_{12}\}, \quad \cal H_{e_5}=\{H_5,H_{10}\}, \quad \cal H_{e_6}=\{H_8,H_{11}\}.
\end{align}
 The $\phi$'s for all small loops in the clockwise direction are equal to $- (2\,a+\Delta)$ and in counter-clockwise direction equal to $-(2\,a+\Delta+\varepsilon)$. The set $\cal H^{(e_1)}$ contains two tree-loop-trees made by removing the edge $\{x,y\}$ from $H_3$ and $H_6$. If we give orientations to both, we get four oriented tree-loop-trees for which the $\phi$ value is either equal to $-(2\,a+\Delta)$ or to $-(2\,a+\Delta+\varepsilon)$. That scenario repeats itself for  $\cal H^{(e_2)}$ and $\cal H^{(e_3)}$. Next, we take an edge connecting the outer triangle to the center and look at the sets $\cal H^{(e_4)}=\cal H^{(e_5)}=\cal H^{(e_6)}$ all containing two elements. Giving orientations to these two tree-loop-trees we get four oriented tree-loop-trees. The $\phi$ of the large loop in clockwise direction is zero and in counter-clockwise direction $-3 \varepsilon$. We see that Condition 1a is satisfied for the edges $e_1, e_2, e_3$ but not for the edges $e_4, e_5, e_6$. In fact,  it is easy to convince oneself that $V(x), V(y)$ and $V(z)$ are diverging for $a>0$ while $V(u)$ is uniformly bounded only when $\Delta\, >\, a$.  That is interesting because it shows that the conditions of our Theorems are not necessary:  for $a<\Delta$ the heat capacity and the excess heat do go to zero at absolute zero.  The reason is that the divergence of the quasipotential
 $V(x)= V(y) = V(z)$ in $\beta\uparrow \infty$ is slower than how their stationary probabilities go to zero.  For $a> \Delta$ also $V(u)$ is diverging and the heat capacity as well.  That again shows the physical content of the condition 1a:  if $a > \Delta$, there is a high barrier between the states $x,y,z$ on the one hand and $z$ on the other hand.  Starting say from state $x$ shows much delay in reaching the dominant state $u$. 
 
 \end{example}

\end{document}